\documentclass[a4paper,UKenglish]{llncs}

\usepackage[firstpage]{draftwatermark}
\SetWatermarkText{\hspace*{5.0in}\raisebox{6.6in}{\includegraphics[scale=0.1]{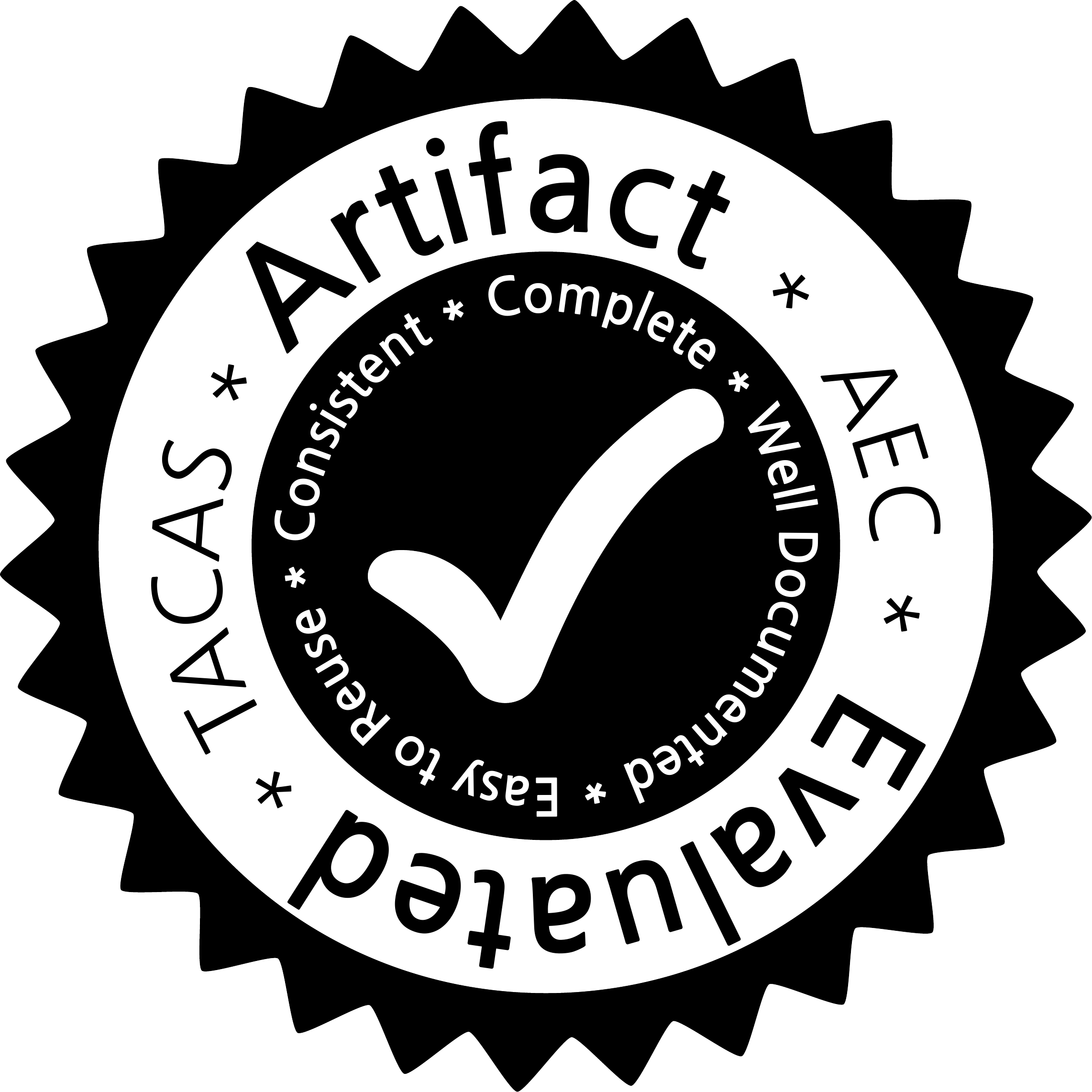}}}
\SetWatermarkAngle{0} 
\usepackage[T1]{fontenc}

\usepackage[utf8]{inputenc}

\usepackage{microtype}
\usepackage{lmodern}

\usepackage{pgfplots}

\usepackage{amssymb,amsmath}

\usepackage{xspace,relsize,stmaryrd}
\usepackage{amssymb,amsmath}
\usepackage{mathpartir}
\usepackage{bussproofs}
\usepackage{multirow}

\allowdisplaybreaks[2]

\usepackage{tikz}
 \usetikzlibrary{trees}
 \usetikzlibrary{shapes}
 \usetikzlibrary{fit}
 \usetikzlibrary{shadows}
 \usetikzlibrary{backgrounds}
 \usetikzlibrary{arrows,automata}

\tikzset{
   n/.style= {circle,fill,inner sep=1.5pt,node distance=2cm}
  ,acc/.style={circle,draw,inner sep=3pt,node distance=2cm}
  ,phantom/.style={circle},
  ,arr/.style={->, >=stealth, semithick, shorten <= 3pt, shorten >= 3pt}
}

\newcounter{blubber}

\makeatletter
\def\moverlay{\mathpalette\mov@rlay}
\def\mov@rlay#1#2{\leavevmode\vtop{%
   \baselineskip\z@skip \lineskiplimit-\maxdimen
   \ialign{\hfil$\m@th#1##$\hfil\cr#2\crcr}}}
\newcommand{\charfusion}[3][\mathord]{
    #1{\ifx#1\mathop\vphantom{#2}\fi
        \mathpalette\mov@rlay{#2\cr#3}
      }
    \ifx#1\mathop\expandafter\displaylimits\fi}
\makeatother

\newcommand{\dotcup}{\charfusion[\mathbin]{\cup}{\cdot}}

\newcommand\pfun{\mathrel{\ooalign{\hfil$\mapstochar\mkern5mu$\hfil\cr$\to$\cr}}}

\newcommand{\al}{\mathsf{al}}
\newcommand{\ad}{\mathsf{ad}}

\newcommand{\defer}{\mathsf{dfr}}

\newcommand{\FLphi}{\mathbf F}

\newcommand\ExpTime{$\textsc{ExpTime}$\xspace}

\newcommand{\Ni}{\noindent}
\newcommand{\Sem}[1]{{[\![#1]\!]}}
\newcommand{\hearts}{\heartsuit}

\newcommand{\sem}[1]{[\![#1]\!]}

\newcommand{\Pow}{\mathcal{P}}

\newcommand{\lO}{\mathcal{O}}
\newcommand{\infrule}[2]{\frac{#1}{#2}}

\newcommand{\ov}[1]{\overline{#1}}

\spnewtheorem{thm}[theorem]{Theorem}{\bfseries}{\itshape}
\spnewtheorem{cor}[theorem]{Corollary}{\bfseries}{\itshape}
\spnewtheorem{cnj}[theorem]{Conjecture}{\bfseries}{\itshape}
\spnewtheorem{lem}[theorem]{Lemma}{\bfseries}{\itshape}
\spnewtheorem{lemdefn}[theorem]{Lemma and Definition}{\bfseries}{\itshape}
\spnewtheorem{prop}[theorem]{Proposition}{\bfseries}{\itshape}
\spnewtheorem{defn}[theorem]{Definition}{\bfseries}{\upshape}
\spnewtheorem{rem}[theorem]{Remark}{\bfseries}{\upshape}
\spnewtheorem{notation}[theorem]{Notation}{\bfseries}{\upshape}
\spnewtheorem{expl}[theorem]{Example}{\bfseries}{\upshape}
\spnewtheorem{thmdefn}[theorem]{Theorem and Definition}{\bfseries}{\itshape}
\spnewtheorem{propdefn}[theorem]{Proposition and Definition}{\bfseries}{\itshape}
\spnewtheorem{assumption}[theorem]{Assumption}{\bfseries}{\upshape}
\spnewtheorem{algorithm}[theorem]{Algorithm}{\bfseries}{\upshape}

 \renewenvironment{theorem}{\begin{thm}}{\end{thm}}
 \renewenvironment{corollary}{\begin{cor}}{\end{cor}}
 \renewenvironment{lemma}{\begin{lem}}{\end{lem}}
 
 \renewenvironment{definition}{\begin{defn}}{\end{defn}}
 
 \renewenvironment{example}{\begin{expl}}{\end{expl}}

\begin{document}

\bibliographystyle{myabbrv}

\title{Permutation Games\\ for the Weakly Aconjunctive $\mu$-Calculus\protect}

\author{Daniel Hausmann \and Lutz Schr\"oder \and Hans-Peter Deifel}
\institute{Friedrich-Alexander-Universit\"{a}t
Erlangen-N\"urnberg, Germany}

\maketitle

\begin{abstract}
  We introduce a natural notion of limit-deterministic parity automata
  and present a method that uses such automata to construct
  satisfiability games for the weakly aconjunctive fragment of the
  $\mu$-calculus. To this end we devise a method that determinizes
  limit-deterministic parity automata of size $n$ with $k$ priorities
  through limit-deterministic Büchi automata to deterministic parity
  automata of size $\mathcal{O}((nk)!)$ and with $\mathcal{O}(nk)$
  priorities.  The construction relies on limit-determinism to
  avoid the full complexity of the Safra/Piterman-construction by
  using partial permutations of states in place of Safra-Trees. By
  showing that limit-deterministic parity automata can be used to
  recognize unsuccessful branches in pre-tableaux for the 
	weakly aconjunctive $\mu$-calculus, we obtain satisfiability
	games of size $\mathcal{O}((nk)!)$ with $\mathcal{O}(nk)$ 
	priorities for weakly
  aconjunctive input formulas of size $n$ and alternation-depth $k$. A
  prototypical implementation that employs a tableau-based global
  caching algorithm to solve these games on-the-fly shows promising
  initial results.
\end{abstract}

\section{Introduction}

The modal $\mu$-calculus~\cite{Kozen83} is an expressive logic for
reasoning about concurrent systems. Its satisfiability problem is
\ExpTime-complete~\cite{EmersonJutla99}. Due to nesting of fixpoints,
the semantic structure of the $\mu$-calculus is quite involved, which
is reflected in the high degree of sophistication of reasoning
algorithms for the $\mu$-calculus.  One convenient modular approach is
the definition of suitable \emph{satisfiability games}
(e.g.~\cite{FriedmannLange13a}); solving such games (i.e. computing
their winning regions) then amounts to deciding the satisfiability of
the input formulas. A standard method for obtaining satisfiability
games is to first construct a \emph{tracking automaton} that accepts
the \emph{bad branches} in a pre-tableau for the input formula, i.e.\
those that infinitely defer satisfaction of a least fixpoint; this
automaton then is determinized and complemented, and the
satisfiability game is built over the carrier set
of the resulting automaton.
The moves in the game are those transitions from
the automaton that correspond to applications of tableau-rules; the existence
of a winning
strategy in this game ensures the existence of a
model, i.e. a locally coherent structure that does not contain bad
branches.  As they typically incur exponential blowup, good
determinization procedures for automata on infinite words play a
crucial role in standard decision procedures for the satisfiability
problem of the $\mu$-calculus and its fragments; in particular, better
determinization procedures lead to smaller satisfiability games which
are easier to solve.

The \emph{weakly aconjunctive} $\mu$-calculus~\cite{Kozen83,Walukiewicz00} restricts
occurrences of recursion variables in conjunctions but is still quite
expressive, e.g.\ can define winning regions in parity games with bounded number of priorities~\cite{DawarGraedel08}. The key
observation for the present paper is that in the weakly aconjunctive case,
pre-tableau branches are made `bad' by a single formula; this implies that
the tracking automaton for such formulas is
\emph{limit-deterministic}, i.e. that is is sufficient to
deterministically track a single formula from some point on. This
motivates a notion of \emph{limit-deterministic parity automata} in
which all accepting runs are deterministic from some point on.
Because the nondeterminism is restricted to finite prefixes of
accepting runs in such automata, they can be determinized in a simpler
way than unrestricted parity automata. We present a reformulation of a
recent determinization method for limit-deterministic \emph{B\"uchi}
automata~\cite{EsparzaKRS17}. The method is inspired by, but
significantly less involved than the more general Safra/Piterman
construction~\cite{Safra88,Piterman07}, essentially due to the fact
that the tree structure of Safra trees collapses, leaving only the
permutation structure. The resulting parity
automaton 
can thus be described as a \emph{permutation automaton}. The method
yields deterministic parity automata with $\mathcal{O}(n!)$
states, compared to $\mathcal{O}((n!)^2)$ in the Safra/Piterman
construction. Crucially, we show that we obtain a similarly simplified
determinization for limit-deterministic \emph{parity} automata by
translating into B\"uchi automata.


As indicated above, limit-deterministic parity automata are able to
recognize bad branches in pre-tableaux for weakly 
aconjunctive $\mu$-calculus formulas. Employing them in the 
standard construction of
satisfiability games, we obtain \emph{permutation games} in which
nodes from the pre-tableau are annotated with a partial permutation
(i.e. a non-repetitive list) of (levelled) formulas. A parity
condition is used to detect indices in the permutation that are active
infinitely often without ever being removed from the permutation. The
resulting parity games are of size $\mathcal{O}((nk)!)$ and have
$\mathcal{O}(nk)$ priorities; as a side result, we thus obtain a new
bound $\mathcal{O}((nk)!)$ on model size for weakly aconjunctive formulas.

The resulting decision procedure generalizes to the weakly
aconjunctive \emph{coalgebraic} $\mu$-calculus, thus covering also,
e.g., probabilistic and alternating-time versions of the
$\mu$-calculus. The generic algorithm has been implemented as an
extension of the \emph{Coalgebraic Ontology Logic Reasoner}
(COOL)~\cite{GorinEA14,HausmannEA16}.  Our implementation constructs
and solves the presented permutation games \emph{on-the-fly}, possibly
finishing satisfiability proofs early, and shows promising initial
results.  The content of the paper is structured as follows: We
describe the determinization of limit-deterministic automata in
Section~\ref{section:det} and the construction of permutation games in
Section~\ref{section:games}, and discuss implementation and evaluation
in Section~\ref{section:cool}.

\medskip\noindent{\textbf{\sffamily Related Work}}\quad Liu and
Wang~\cite{LiuWang09} give a tighter estimate $\mathcal{O}((n!)^2)$
for the number of states in Piterman's
determinization~\cite{Piterman07}. Schewe~\cite{Schewe09} simplifies
Piterman's construction (establishing the same bound as Liu and Wang).
Tian and Duan~\cite{TianDuan14} further improve Schewe's
construction. Fisman and Lustig~\cite{FismanLustig15} present a
modularization of B\"uchi determinization that is aimed mainly at
easing understanding of the construction. Parity automata can be
determinized by first converting them to B\"uchi automata and then
applying B\"uchi determinization. Schewe and
Varghese~\cite{ScheweVarghese14} address the direct determinization of
parity automata (via Rabin automata), and prove optimality within a
small constant factor, and even absolute optimality for the B\"uchi
subcase. All these constructions and estimates concern unrestricted
B\"uchi or parity automata. Recently, Safra-less determinization of
limit-deterministic Büchi automata has been described in the context
of controller synthesis for LTL~\cite{EsparzaKRS17}; the
determinization method that we present in Section 2.2.\ has been
devised independently from~\cite{EsparzaKRS17} but employs a very
similar construction (yielding essentially the same results on the
complexity of the construction).

The use of games in $\mu$-calculus satisfiability checking goes back
to Niwi\'nski and Walukiewicz~\cite{NiwinskiWalukiewicz96} and has
since been extended to the
unguarded $\mu$-calculus~\cite{FriedmannLange13a} and 
the \emph{coalgebraic}
$\mu$-calculus~\cite{CirsteaEA09}. Game-based procedures for the
relational $\mu$-calculus have been implemented in
MLSolver~\cite{FriedmannLange10MLSolver}, and for the alternation-free
coalgebraic $\mu$-calculus in COOL~\cite{HausmannEA16}.

\section{Determinizing Limit-Deterministic Automata}\label{section:det}

\subsection{Limit-deterministic automata}
We recall the basics of parity automata: A \emph{parity automaton} is
a tuple $\mathcal{A}=(V,\Sigma,\delta,u_0,\alpha)$ where $V$ is a set
of \emph{states}, $\Sigma$ is an \emph{alphabet},
$\delta\subseteq V\times\Sigma\times V$ is a \emph{transition relation}, $u_0\in V$ is
an \emph{initial state}, and $\alpha:\delta\to\mathbb{N}$ is a
\emph{priority function} that assigns natural numbers to \emph{transitions}
(assigning priorities to transitions rather than states yields a slightly more
succinct notion of automata while retaining the computational properties of standard
parity automata~\cite{ScheweVarghese14}). For $(v,a)\in V\times\Sigma$, we write $\delta(v,a)=\{u\mid (v,a,u)\in\delta\}$. The
\emph{index} $\mathsf{idx}(\mathcal{A})=\max\{\alpha(t)\mid t\in \delta\}$
of a parity automaton $\mathcal{A}$ is its maximal priority.  A
\emph{run} $\rho=v_0v_1\ldots$ of $\mathcal{A}$ on an infinite word
$w=a_0 a_1\ldots\in \Sigma^\omega$ starting at $v\in V$ is a (possibly
infinite) sequence of states $v_i$ such that $v_0=v$ and for all
$i\geq 0$, $v_{i+1}\in\delta(v_i,a_i)$.  We see runs~$\rho$ or words~$w$ as functions from natural numbers to states $\rho(i)=v_i\in V$ or
letters $w(i)=a_i\in\Sigma$. For a run
$\rho$ on a word $w$, we define the according sequence $\mathsf{trans}(\rho)$ of transitions by
$\mathsf{trans}(\rho)(i)=(\rho(i),w(i),\rho(i+1))$.
We denote the set of all runs of
$\mathcal{A}$ on a word $w$ starting at $v$ by
$\mathsf{run}(\mathcal{A},v,w)$, or just by
$\mathsf{run}(\mathcal{A},w)$ if $v=u_0$.  A run~$\rho$ of
$\mathcal{A}$ on a word $w$ is \emph{accepting} if the highest priority that occurs
infinitely often in it (notation:
$\max(\mathsf{Inf}(\alpha\circ\mathsf{trans}(\rho)))$; 
we generally write $\mathsf{Inf}(s)$ for the set of elements occurring infinitely
often in a sequence~$s$) is even.  A parity automaton
$\mathcal{A}$ \emph{accepts} an infinite word $w$ if
$\mathsf{run}(\mathcal{A},w)$ contains an accepting run, and we denote
by $L(\mathcal{A})\subseteq \Sigma^\omega$ the set of all words that
are accepted by $\mathcal{A}$.  

Given a state $v\in V$ and a letter $a\in\Sigma$, we define
$\delta|_{v,a}=\{(v,a,u)\mid u\in\delta(v,a)\}$.
Given a set $\gamma\subseteq \delta$ of transitions,
a state $v\in V$, a set of states $U\subseteq V$ and a letter $a\in\Sigma$,
we put $\gamma(U,a)=\bigcup\{\gamma(v,a)\mid v\in U\}$;
given a finite word
$w=a_0\ldots a_n$, we then recursively define
$\gamma(v,w)=\gamma(\gamma(v,a_0),a_1\ldots a_n)$, obtaining the
set of all states reachable from $v$ when reading $w$ while only
using transitions from $\gamma$. For $U\subseteq V$, $\gamma\subseteq \delta$ and $w\in\Sigma^*$, we
put $\gamma(U,w)=\bigcup\{\gamma(u,w)\mid u\in U\}$.
Furthermore, we define the set of states that are \emph{reachable}
from a node $v\in V$ using transitions from $\gamma$ as
$\mathsf{reach}_\gamma(v)=\bigcup\{\gamma(v,w)\mid w\in\Sigma^*\}$;
we extend this notation to sets of nodes, putting
$\mathsf{reach}_{\gamma}(U)=\bigcup\{\mathsf{reach}_{\gamma}(u)\mid u\in U\}$ for $U\subseteq V$.
If $\gamma=\delta$, then we omit the subscripts.  A state
$v\in V$ is said to be \emph{deterministic} (in $\gamma\subseteq \delta$)
if it has at most one ($\gamma$-)successor for each
letter $a\in \Sigma$. A set $U\subseteq V$ is deterministic (in
$\gamma\subseteq \delta$) if every state $v\in U$ is deterministic (in $\gamma$).
The automaton $\mathcal{A}$ is said to be \emph{deterministic} if $V$
is deterministic; the transition relation in deterministic automata
hence is a partial function (since such automata can be transformed to
equivalent automata with total transition function, this definition
suffices for purposes of determinization).  We put
$\alpha(i)=\{t\in \delta\mid \alpha(t)= i\}$ and
$\alpha_\leq(i)=\{t\in \delta\mid \alpha(t)\leq i\}$. 

 A \emph{B\"uchi
  automaton} is a parity automaton with only the priorities $1$ and
$2$; the set of \emph{accepting transitions} then is
$F=\alpha(2)$ and a run is accepting if it passes infinitely many
accepting transitions.
 For B\"uchi automata, we assume
w.l.o.g. that every transition $t\in F$ is part of a cycle.
We use the abbreviations (N/D)PA, (N/D)BA to denote the different types of automata.

Our notion of limit-determinism of automata is defined as a semantic property:

\begin{definition}[Limit-deterministic parity automata]
A PA $\mathcal{A}=(V,\Sigma,\delta,u_0,\alpha)$ is \emph{limit-deterministic}
if there is, for each word $w$ and each accepting run $\rho\in\mathsf{run}(\mathcal{A},w)$, 
a number $i$ such that for all $j\geq i$, 
$\delta|_{\rho(j),w(j)}\cap \alpha_\leq(l)=\{\mathsf{trans}(\rho)(j)\}$, where $l=\max(\mathsf{Inf}(\alpha\circ\mathsf{trans}(\rho)))$.
\end{definition}
If $\mathcal{A}$ is a BA, then 
we have $\max(\mathsf{Inf}(\alpha\circ\mathsf{trans}(\rho)))=2$
for every accepting run $\rho$; as $\alpha_\leq(2)=\delta$, the above definition
instantiates to requiring the existence of a number~$i$ such that for all $j\geq i$, $\delta(\rho(j),w(j))=\{\rho(j+1)\}$.

\begin{definition}[Compartments]
Given a PA $\mathcal{A}=(V,\Sigma,\delta,u_0,\alpha)$ with $k$ priorities,
and an even number $l\leq k$, the \emph{$l$-compartment} $C_l(t)$ of a transition $t\in\alpha(l)$
is the set $\mathsf{reach}_{\alpha_\leq(l)}(\pi_3(t))$ where $\pi_3$ projects transitions $t=(v,a,u)$ to their target nodes~$u$.
If $l$ is irrelevant, then we refer to $l$-compartments
just as compartments. The \emph{size} of a compartment $C$ is just $|C|$.
A compartment $C$ is \emph{internally deterministic} if for each $v\in C$ and all $a\in \Sigma$, $|\delta(v,a)\cap C|\leq 1$.
\end{definition}
Note that the union of all $l$-compartments is
$\mathsf{reach}_{\alpha_\leq(l)}(\pi_3[\alpha(l)])$.
 Compartments allow for a syntactic characterization of limit-determinism:

\begin{lemma}\label{lem:evdet}
A PA is limit-deterministic 
if and only if all its compartments are internally deterministic.
\end{lemma}
\begin{corollary}\label{lem:evdetdec}
  It is decidable in polynomial time whether a given automaton is limit-deterministic.
\end{corollary}
Lemma~\ref{lem:evdet} specializes to BA as follows: we have
$\alpha(0)=\emptyset$, $\alpha_\leq(2)=\delta$ and 
$\alpha(2)=F$, so that
the union of all $0$-compartments is empty and that of all
$2$-compartments is $\mathsf{reach}(\pi_3[F])$; thus a BA is
limit-deterministic if and only if $\mathsf{reach}(\pi_3[F])$ is deterministic. 
Such B\"uchi automata are also called
\emph{semi-deterministic}~\cite{CourcoubetisY95}.

\subsection{Determinizing Limit-Deterministic B\"uchi Automata}

The Safra/Piterman construction~\cite{Safra88,Piterman07} determinizes
B\"uchi automata by means of so-called Safra trees, i.e.  trees whose
nodes are labelled with sets of states of the input automaton such
that the label of a node is a proper superset of the union of all its
children's labels. Additionally, the nodes are ordered by their age
and upon each transition between Safra trees, the ages of the oldest
nodes that are active and/or removed during this transition
determine the priority of the new Safra tree.  In its original
formulation, the Safra/Piterman construction adds new child nodes to
the graph that are labelled with the accepting states in their
parent's label. We observe that this step can be modified slightly --
without affecting the correctness of the construction -- by letting
every accepting state from the parent's label receive its own separate
child node; then the labels of newly created nodes are always
singletons.  Limit-determinism of the input automaton then implies
that the node labels also \emph{remain} singletons. Since singleton
nodes do not have children in Safra trees, this leads to the collapse
of their tree structure; the resulting data structure is essentially a
partial permutation, i.e. a non-repetitive list, of states (ordered by
their age).  The arising modified Safra/Piterman construction for the
limit-deterministic case boils down to the following method, which a)
has a relatively short presentation and a simpler correctness proof
than the full Safra/Piterman construction, and b) results in asymptotically
smaller automata; the underlying idea of the construction has first
been described in the context of controller synthesis for
LTL~\cite{EsparzaKRS17}.

\begin{definition}[Partial permutations]
  Given a set $U$ of states, let $\mathsf{pperm}(U)$ denote the set of
  \emph{partial permutations} over $U$, i.e.  the set of
  non-repetitive lists $l=[v_1,\ldots,v_n]$ with $v_i\neq v_j$ for
  $i\neq j$ and $v_i\in U$, for all $1\leq i\leq n$. We denote the
  $i$-th element in $l$ by $l(i)=v_i$,
	the empty partial permutation by $[\,]$ and the
  length of a partial permutation $l$ by $|l|$.
\end{definition}

\begin{definition}[Determinization of limit-deterministic BA]\label{defn:permdet}
Fix a limit-deterministic BA $\mathcal{A}=(V,\Sigma,\delta,u_0,F)$, and put $Q=\mathsf{reach}(\pi_3[F])$, $\ov{Q}=V\setminus Q$, $q=|Q|$. Define the DPA 
$\mathcal{B}=(W,\Sigma,\delta',w_0,\alpha)$ by putting $W=\mathcal{P}(\ov{Q})\times \mathsf{pperm}(Q)$,
$w_0 = (\{u_0\},[\,])$ if $u_0\in \overline{Q}$, $w_0=(\emptyset,[u_0])$ if $u_0\in Q$
and for $g=(U,l)\in W$ and $a\in\Sigma$, $\delta'(g,a)=h$, where
$h=(\delta(U,a)\cap \ov{Q},l')$ and
where $l'$ is constructed from $l=[v_1,\ldots,v_m]$ as follows:
\begin{enumerate}
\item Define a list $t$ of length $m$ over $Q\cup\{*\}$ (with $*$
  representing undefinedness) in which $t(i)=w$ if
  $\delta(v_i,a)=\{w\}$, and $t(i)=*$ if $\delta(v_i,a)=\emptyset$.
\item For $j<k$ and $t(j)=t(k)$, put $t(k)=*$.
\item
Remove undefined entries in $t$, formally: for each $1\leq i\leq |t|$,
if $t(i)=*$, then iteratively put $t(j)=t(j+1)$ for each $i\leq j\leq |t|$, starting at $i$.
\item
For any $w\in\delta(U,a)\cap Q$ that does not occur in $t$, add $w$ to the end of $t$.
If there are several such $w$, the order in which they are added to $t$ is irrelevant.
\item Put $l'=t$.
\end{enumerate}
Temporarily, $t$ may contain duplicate or undefined entries, but Steps~2.\ and~3.\
ensure that in the end, $t$ is a partial permutation of length at most $q$.  Let $r$ (for
`removed') denote the lowest index $i$ such that $t(i)=*$ after
Step~2.  Let $a$ (for `active') denote the lowest index $i$ such that
$(l(i),a,l'(i))\in F$. If $r>|l'|$ and there is no $i$ with $(l(i),a,l'(i))\in F$, then put
$\alpha(g,a,h)=1$.
Otherwise, put
\begin{align*}
\alpha(g,a,h)=\begin{cases}
2(q-r)+3&\text{if }r\leq a\\
2(q-a)+2&\text{if }r>a.
\end{cases}
\end{align*}

\end{definition}

\begin{theorem}\label{thm:edba}
We have $L(\mathcal{A})=L(\mathcal{B})$, and $\mathcal{B}$ has at most $2n+1$ priorities; for $n\geq 4$,
we have $|W| \leq n!e$.
\end{theorem}

\begin{corollary}\label{cor:edba}
Limit-deterministic B\"uchi automata of size $n$ 
can be determinized to deterministic parity automata of size 
$\lO(n!)$ and with $\lO(n)$ priorities.
\end{corollary}

\begin{example} Consider the limit-deterministic BA $\mathcal{A}$
  depicted below and the determinized DPA $\mathcal{B}$ that is
  constructed from it by applying the method.  We see by
  Lemma~\ref{lem:evdet} that $\mathcal{A}$ is really
  limit-deterministic: we have $F=\{(1,b,3)\}$, i.e. the
  $b$-transition from state $1$ to state $3$ (depicted with a boxed
  transition label) is the only accepting transition; thus we have
  $Q=\mathsf{reach}(\pi_3[F])=\{1,3\}$ (so $\ov{Q}=\{0,2\}$), and the states
  $1$ and $3$ are deterministic.
  Moreover, $L(\mathcal{A})=L(\mathcal{B})=a(a|b)^+(a^+ b)^\omega$.\\

\begin{minipage}{.3\linewidth}
$\quad$
		 $\mathcal{A}$:\\
\tikzset{every state/.style={minimum size=15pt}}
\begin{tiny}
  \begin{tikzpicture}[
		auto,
    node distance=0.8cm,
    semithick
    ]
     \node[state,initial above] (0) {$0$};
     \node (yo) [below of=0] {};
     \node[state] (1) [left of=yo] {$1$};
     \node[state] (2) [right of=yo] {$2$};
     \node[state] (3) [below of=yo] {$3$};
     \path[->] (0) edge [loop right] node [right] {$a$} (0);
     \path[->] (0) edge node [pos=0.3,left] {$a$} (1);
     \path[->] (0) edge node [pos=0.3,right] {$a$} (2);
     \path[->] (1) edge [loop left] node [left] {$a$} (1);
     \path[->] (1) edge [bend right=30] node [pos=0.3,below] {$\fbox{b}$\;\;\;\;\;} (3);
     \path[->] (2) edge [loop right] node [right] {$a,b$} (2);
     \path[->] (2) edge node [pos=0.6,right] {$a,b$} (3);
     \path[->] (3) edge [bend right=30] node [pos=0.6,right] {$a$} (1);
     
  \end{tikzpicture}
\end{tiny}

    \end{minipage}%
		$\quad$
		$\quad$
		$\quad$
    \begin{minipage}{.7\linewidth}
$\mathcal{B}$:
\tikzset{every state/.style={minimum size=15pt}}
\begin{tiny}

  \begin{tikzpicture}[
		auto,
    node distance=1.0cm,
    semithick
    ]
     \node[rounded corners, draw,initial above] (1) {$\{0\},[\,]$};
		 \node (yo) [right of=1] {};
     \node[rounded corners, draw] (2) [below of=1] {$\emptyset,[\,]$};
     \node[rounded corners, draw] (3) [right of=yo] {$\{0,2\},[1]$};
     \node[rounded corners, draw] (4) [below of=3] {$\{0,2\},[1,3]$};
		 \node (yo2) [right of=3] {};
     \node[rounded corners, draw] (5) [right of=yo2] {$\{2\},[3]$};
     \node[rounded corners, draw] (6) [below of=5] {$\{2\},[1,3]$};
		 \path[->] (1) edge node [left] {$b,1$} (2);
     \path[->] (1) edge node [above] {$a,1$} (3);
     \path[->] (2) edge [loop below] node [below] {$a,b,1$} (2);
     \path[->] (3) edge node [left] {$a,1$} (4);
     \path[->] (3) edge node {$b,4$} (5);
     \path[->] (4) edge [loop below] node [below] {$a,3$} (4);
     \path[->] (4) edge node [pos=0.4,above] {$b,4\,$} (5);
     \path[->] (5) edge [loop above] node [above] {$b,5$} (5);
     \path[->] (5) edge [bend right=20] node [pos=0.6,left] {$a,1$} (6);
     \path[->] (6) edge [loop below] node [below] {$a,3$} (6);
     \path[->] (6) edge [bend right=20] node [pos=0.5, right] {$b,4$} (5);
     \end{tikzpicture}
\end{tiny}
\vspace{10pt}
\end{minipage}
Notice that in $\mathcal{B}$, 
there is a $b$-transition with priority 1 from the initial state to the sink state $(\emptyset,[\,])$
and an $a$-transition to $(\{0,2\},[1])$; as $1\in Q$ but $(0,a,1)\notin F$, this transition has priority 1.
A further $b$-transition leads from $1$ to $3$ in $\mathcal{A}$; in $\mathcal{B}$, we have a $b$-transition
from $(\{0,2\},[1])$ to $(\{2\},[3])$ and since $(1,b,3)\in F$, the first position in the permutation component is active
during this transition so that the transition has priority 4.
Yet another $b$-transition loops from $(\{2\},[3])$ to $(\{2\},[3])$. Since there is no $b$-transition
starting at state $3$, the first element in the permutation is removed in Step~1.\ of the construction. Since there is a
$b$-transition from $2$ to $3$, it is added to the permutation again in Step~4.\ of the construction. Crucially, however, the priority of the
transition is $5$, since the first item of the permutation has been (temporarily) removed. 
The intuition is that the trace of $3$ ends when the letter $b$ is read;
even though a new trace of $3$ immediately starts, we do not consider it to be the same trace as the previous
one. Thus the transition obtains priority $5$ so that it may be used only finitely often in an accepting
run of $\mathcal{B}$, i.e. accepting runs contain an uninterrupted trace that visits state $3$ 
infinitely often. Thus two or more consecutive $b$'s can only occur finitely often in any accepted word.

\end{example}

\subsection{Determinizing Limit-Deterministic Parity Automata}

To determinize limit-deterministic PA, it suffices to transform them to equivalent
limit-deterministic BA and determinize the BA.
This transformation from PA to BA is achieved by a construction which is 
inspired by Theorems 2 and 3 in~\cite{KKV01}; we add the observation that the construction
preserves limit-determinism.

\begin{definition}\label{defn:edpatoedba}
Given a limit-deterministic PA $\mathcal{C}=(V,\Sigma,\delta,u_0,\alpha)$ with $n=|V|$
and $k>2$ priorities,
we define the limit-deterministic BA $\mathcal{D}=(W,\Sigma,\delta',u_0,F)$ by putting
$W=V\cup (V\times \{0,\ldots, \big\lceil{\frac{k-1}{2}}\big\rceil\})$,  
and for $w\in W$ and $a\in\Sigma$,
\begin{align*}
\delta'(v,a)&=\begin{cases}\{(w,m)\mid (v,a,w)\in\alpha(2m)\}\cup \delta(v,a) & \text{if }v\in V\\
														   \{(w,l)\mid (v',a,w)\in\alpha_\leq(2l)\} & \text{if }v=(v',l)\notin V
                   \end{cases}
\end{align*}
Finally, we put $F=\{((v,l),a,(w,l))\in \delta'\mid \alpha(v,a,w)=2l\}$.
To see that~$\mathcal{D}$ is limit-deterministic, it suffices by Lemma~\ref{lem:evdet} to
show that $\mathsf{reach}(\pi_3[F])$ is deterministic. We observe that for each state $(w,l)\in \mathsf{reach}(\pi_3[F])$,
$(w,l)$ is deterministic by definition of $\delta'$ since $w$ is contained in a (by Lemma~\ref{lem:evdet}, internally deterministic)
$2l$-compartment of~$\mathcal{C}$. 
\end{definition}


\begin{lemma}\label{lem:npanba}
We have $L(\mathcal{C})=L(\mathcal{D})$ and $|W|\leq n(\big\lceil{\frac{k}{2}}\big\rceil+1)\leq nk$.
\end{lemma}
By Theorem~\ref{thm:edba}, $\mathcal{D}$ can be determinized to a DPA $\mathcal{E}$
of size at most $(nk)!e$, with at most $nk+2$ priorities
and with $L(\mathcal{D})=L(\mathcal{E})$.

\begin{corollary}
Limit-deterministic parity automata of size $n$ with $k$ priorities
can be determinized to deterministic parity automata of size 
$\lO((nk)!)$ and with $\lO(nk)$ priorities.
\end{corollary}



\section{Permutation Games for the
aconjunctive $\mu$-Calculus}\label{section:games}

\subsection{The $\mu$-Calculus}
\Ni We briefly recall the definition of the 
$\mu$-calculus. We fix a set $P$ of \emph{propositions}, a set $A$ of
\emph{actions}, and a set $\mathfrak{V}$ of fixpoint
variables. The set $\mathsf{L}_\mu$ of $\mu$-calculus formulas is the
set of all formulas $\phi,\psi$ that can be constructed by the grammar
\begin{align*}
\psi,\phi ::= \bot \mid \top \mid p \mid \neg p \mid X \mid \psi\wedge\phi \mid \psi\vee\phi
\mid \langle a \rangle\psi \mid [a]\psi \mid \mu X.\, \psi \mid \nu X.\, \psi
\end{align*}
where $p\in P$, $a\in A$, and $X\in \mathfrak{V}$; we write $|\psi|$
for the size of a formula $\psi$. Throughout the paper, we use $\eta$
to denote one of the fixpoint operators $\mu$ or $\nu$. We refer to
formulas of the form $\eta X.\,\psi$ as \emph{fixpoint literals}, to
formulas of the form $\langle a\rangle\psi$ or $[a]\psi$ as
\emph{modal literals}, and to $p$, $\neg p$ as \emph{propositional
  literals}. The operators $\mu$ and $\nu$ \emph{bind} their
variables, inducing a standard notion of \emph{free variables} in
formulas. We denote the set of free variables
of a formula $\psi$ by $\mathsf{FV}(\psi)$. A formula $\psi$ is
\emph{closed} if $\mathsf{FV}(\psi)=\emptyset$, and \emph{open}
otherwise. We write $\psi\leq\phi$ ($\psi<\phi$) to indicate that
$\psi$ is a (proper) subformula of $\phi$. We say that $\phi$
\emph{occurs free} in $\psi$ if $\phi$ occurs in $\psi$ as a
subformula that is not in the scope of any fixpoint operator.
Throughout, we \emph{restrict to formulas that are guarded}, i.e.\
have at least one modal operator between any occurrence of a variable
$X$ and an enclosing binder $\eta X$.  (This is standard although
possibly not without loss of generality~\cite{FriedmannLange13a}.)
Moreover we assume w.l.o.g.\ that input formulas are \emph{clean},
i.e.\ all fixpoint variables are mutually distinct and distinct from all free variables, and \emph{irredundant},
i.e.\ $X\in \mathsf{FV}(\psi)$ for all subformulas $\eta X.\,\psi$.
We refer to a variable~$X$ that is bound by a least (greatest)
fixpoint operator $\mu X.\chi$ ($\nu X.\chi$) in a formula~$\phi$ as a \emph{$\mu$-variable} (\emph{$\nu$-variable}) of~$\phi$, and to the
process of substituting such an~$X$ with its binding
fixpoint literal ($\mu X.\chi$ or $\nu X.\chi$, respectively) as \emph{unfolding}. An occurrence of a subformula $\psi$ of a formula
$\phi$ \emph{contains an active $\mu$-variable}~\cite{Kozen83} if~$\psi$ can be converted into a formula containing a free occurrence of
a $\mu$-variable of~$\phi$ by repeatedly unfolding $\nu$-variables of~$\phi$.  

Formulas are evaluated over \emph{Kripke structures}
$\mathcal{K}=(W,(R_a)_{a\in A},\pi)$, consisting of a set $W$ of
\emph{states}, a family $(R_a)_{a\in A}$ of relations
$R_a\subseteq W\times W$, and a valuation $\pi:P\to\Pow(W)$ of the
propositions. Given an \emph{interpretation}
$i:\mathfrak{V}\to \Pow(W)$ of the fixpoint variables, define
$\Sem{\psi}_i\subseteq W$ by the obvious clauses for Boolean operators
and propositions,
$\sem{X}_i=i(X)$, 
$\sem{\langle a\rangle\psi}_i=\{v\in W\mid \exists w\in R_a(v).w\in
\sem{\psi}_i\}$,
$\sem{[a]\psi}_i=\{v\in W\mid \forall w\in R_a(v).w\in\sem{\psi}_i\}$,
$\sem{\mu X.\,\psi}_i =\mu\sem{\psi}^X_i$ and
$\sem{\nu X.\,\psi}_i =\nu\sem{\psi}^X_i$, where
$R_a(v)=\{w\in W\mid (v,w)\in R_a\}$,
$\sem{\psi}^X_i(G) = \sem{\psi}_{i[X\mapsto G]}$, and $\mu$, $\nu$
take least and greatest fixpoints of monotone functions, respectively.
If $\psi$ is closed, then $\sem{\psi}_i$ does not depend on $i$, so we
just write $\sem{\psi}$. We denote the \emph{Fischer-Ladner closure}~\cite{Kozen88} of a formula $\phi$
by $\mathbf{F}(\phi)$, or just by $\mathbf{F}$, if no confusion arises; intuitively, $\mathbf{F}$
is the set of formulas that
can arise as subformulas when unfolding each fixpoint operator in $\phi$ at most once. We note 
$\mathbf{F}\leq|\phi|$~\cite{Kozen88}.

The \emph{aconjunctive fragment}~\cite{Kozen83} of the $\mu$-calculus
is obtained by requiring that for all conjunctions that occur as a
subformula, at most one of the conjuncts contains an active
$\mu$-variable.  In the \emph{weakly aconjunctive
  fragment}~\cite{Walukiewicz00}, this requirement is loosened to the
constraint that all conjunctions that occur as a subformula and
contain an active $\mu$-variable are of the shape
$\psi\wedge \Diamond \psi_1\wedge\ldots
\wedge\Diamond\psi_n\wedge\Box(\psi_1\vee\ldots\vee\psi_n)$,
where $\psi$ does not contain active $\mu$-variables. For instance,
for all $n$, the formula
$\eta X_n\ldots\mu X_1. \nu X_0. \bigvee_{0\leq i\leq n}
(q_i\wedge\Diamond X_i)$
is aconjunctive (and equivalent to the weakly aconjunctive formula
obtained by replacing $\Diamond X_i$ with
$\Diamond X_i\land\Diamond\top\land\Box(X_i\lor\top)$). The permutation
satisfiability games that we introduce work for the more expressive
weakly aconjunctive fragment.

We will make use of the standard \emph{tableau rules}~\cite{FriedmannLange13a} (each consisting of one
\emph{premise} and a possibly empty set of \emph{conclusions}):
		\begin{align*}
  (\bot)\quad & \;\;\;\;\quad\quad\infrule{\Gamma,\bot}
	{}
&  (\lightning)\quad & \quad\quad\infrule{\Gamma,p,\neg p}
	{}
&
  (\wedge)\quad & \;\quad\quad\infrule{\Gamma,\psi\wedge \phi}
	{\Gamma,\psi,\phi}
 \\[5pt] (\vee) \quad & \quad\infrule{\Gamma,\psi\vee \phi}
	{{\Gamma,\psi}\qquad {\Gamma,\phi}}
&
  (\langle a\rangle) \quad & \infrule{\Gamma,[a] \psi_1,\ldots,[a]\psi_n,\langle a\rangle \phi}
	{\psi_1,\ldots,\psi_n,\phi}
  &(\eta) \quad &\infrule{\Gamma,\eta X.\, \psi}
	{\Gamma,\psi [X\mapsto \eta X.\, \psi]}
\end{align*}
(for $a\in A$, $p\in P$); we refer to the tableau rules by $\mathcal{R}$ and usually write rule applications
with premise $\Gamma$ and conclusion $\Sigma=\Gamma_1,\ldots,\Gamma_n$ sequentially: $(\Gamma/\Sigma)$.

To track fixpoint formulas through pre-tableaux,
we will use deferrals, that is, the decomposed form of
formulas that are obtained by unfolding fixpoint
literals.
\begin{definition}[Deferrals]\label{defn:affil}
  Given fixpoint literals $\chi_i = \eta X_i.\,\psi_i$, $i=1,\dots,n$,
  we say that a substitution
  $\sigma=[X_1\mapsto \chi_1];\ldots;[X_n\mapsto\chi_n]$ \emph{sequentially
    unfolds $\chi_n$} if $\chi_i <_f \chi_{i+1}$ for all $1\leq i<n$,
  where we write $\psi <_f \eta X.\,\phi$ if $\psi\leq\phi$ and $\psi$
  is open and occurs free in $\phi$ (i.e.\ $\sigma$ unfolds a nested
  sequence of fixpoints in $\chi_n$ innermost-first).  We say that a
  formula $\chi$ is \emph{irreducible} if for every substitution
  $[X_1\mapsto \chi_1];\ldots;[X_n\mapsto \chi_n]$ that sequentially unfolds
  $\chi_n$, we have that
  $\chi = \chi_1([X_2\mapsto \chi_2];\ldots;[X_n\mapsto \chi_n])$
  implies $n=1$ (i.e.\ $\chi=\chi_1$). A formula $\psi$
  \emph{belongs} to an irreducible closed fixpoint literal $\theta_n$, or is a
  \emph{$\theta_n$-deferral}, if $\psi=\alpha\sigma$ for some
  substitution
  $\sigma = [X_1\mapsto \theta_1];\ldots;[X_n\mapsto \theta_n]$ that
  sequentially unfolds $\theta_n$ and some $\alpha <_f \theta_1$.  We denote
  the set of $\theta_n$-deferrals by $\defer(\theta_n)$.
\end{definition}
\Ni E.g.\ the substitution
$\sigma=[Y\mapsto \mu Y.\,(\Box X\land\Diamond\Diamond
Y)];[X\mapsto\theta]$ 
sequentially unfolds the irreducible closed formula
$\theta=\nu X.\,\mu Y.\,(\Box X\land\Diamond\Diamond Y)$, and
$(\Diamond Y)\sigma=\Diamond\mu Y.\,(\Box\theta\land\Diamond\Diamond
Y)$
is a $\theta$-deferral. A fixpoint literal is irreducible if it is not
an unfolding $\psi[X\mapsto\eta X.\,\psi]$ of a fixpoint literal
$\eta X.\,\psi$; in particular, every clean irredundant fixpoint
literal is irreducible. 


As a technical tool, we define a measure for the depth 
of alternation at which a deferral resides inside the fixpoint to which it belongs:

\begin{definition}[Alternation level and alternation depth]
The \emph{alternation level} $\al(\phi\sigma):=\al(\sigma)$ of a deferral $\phi\sigma$
is defined inductively over $|\sigma|$, where $\al(\epsilon)=\al(\epsilon)_\mu=\al(\epsilon)_\nu=0$, for the
empty substitution $\epsilon$,
$\al(\sigma;[X\mapsto \eta X.\,\psi])=\al(\sigma)_\mu+1$ if $\eta=\mu$ and $\al(\sigma;[X\mapsto \eta X.\,\psi])=\al(\sigma)_\nu$ otherwise, and
\begin{small}
\begin{align*}
\al(\sigma;[X\mapsto \eta X.\,\psi])_\mu=\begin{cases}\al(\sigma)_\mu&\text{if }\eta=\mu\\\al(\sigma)_\nu+1&\text{otherwise}\end{cases} \\
\al(\sigma;[X\mapsto \eta X.\,\psi])_\nu=\begin{cases}\al(\sigma)_\nu&\text{if }\eta=\nu\\\al(\sigma)_\mu+1&\text{otherwise}\end{cases}
\end{align*}
\end{small}%
This definition assigns greater numbers to inner fixpoint literals, i.e. to
deferrals which occur at higher nesting depth, i.e. with more alternation inside their sequence $\sigma$.
Given a formula $\psi$, its \emph{alternation depth} $\ad(\phi)$ is defined as
$\ad(\phi)= \max\{\al(\delta)\mid \delta\in\FLphi,\exists\theta.\delta\in\defer(\theta)\}$.

\end{definition}

\subsection{Limit-Deterministic Tracking Automata}

As a first step towards deciding the satisfiability of
a weakly aconjunctive $\mu$-calculus formula $\phi$,
we now construct a tracking automaton that takes 
branches of (that is, infinite paths
through) standard pre-tableaux for $\phi$ as input
and accepts a branch if and only if it contains a 
least fixpoint formula whose satisfaction
is deferred indefinitely on that branch.
To this end, we import the following notions
of threads and tableaux 
from~\cite{FriedmannLange13a}:

\begin{definition}
A \emph{pre-tableau} for a formula $\phi$ is a graph the nodes of which are labelled with subsets of the Fischer-Ladner closure $\FLphi$; the graph structure $L$ of a pre-tableau is constructed by applying tableau rules from $\mathcal{R}$ to the labels of nodes with the requirement that for
each rule application $(\Gamma/\Sigma)$ to the label 
$\Gamma$ of a node $v$, there is a $w$ with 
$(v,w)\in L$ such that the label of $w$ is contained in~$\Sigma$.
Nodes whose labels are \emph{saturated} (i.e.\ do not contain propositional or fixpoint operators)
are called \emph{states}.
Formulas are tracked through rule applications by the \emph{connectedness relation} $\leadsto \subseteq (\Pow(\FLphi)\times \FLphi)^2$ that is defined by putting $\Phi,\phi\leadsto\Psi,\psi$ if and only if $\Psi$ is a conclusion of an application of
a rule from $\mathcal{R}$ to $\Phi$ such that $\phi\in\Phi$, $\psi\in\Psi$, and the rule application transforms $\phi$ to $\psi$;
if the rule application does not change $\phi$, then $\phi=\psi$. E.g. we have $\Phi,\psi_1\wedge\psi_2\leadsto\Psi,\psi_i$,
where $i\in\{1,2\}$ and~$\Psi$ is obtained from~$\Phi$ by applying the rule $(\wedge)$ to $\psi_1\wedge\psi_2$.
A \emph{branch} $\Psi_0,\Psi_1\ldots$ in a pre-tableau is a sequence of labels such that for all $i>0$, $\Psi_{i+1}$
is an $L$-successor of $\Psi_i$. A \emph{thread} on an infinite branch 
$\Psi_0,\Psi_1,\ldots$ is an infinite sequence $t=\psi_0,\psi_1\ldots$ of formulas  with 
$\Psi_0,\psi_0\leadsto\Psi_1,\psi_1\leadsto\ldots$.
A \emph{$\mu$-thread} is a thread $t$ such that $\min(\mathsf{Inf}(\al \circ t))$ is odd,
i.e. the outermost fixpoint literal that is unfolded infinitely often in $t$ is a least
fixpoint literal.
A \emph{bad branch} is an infinite branch that contains a $\mu$-thread.
A \emph{tableau} for $\phi$ is a pre-tableau for $\phi$ that does not contain bad branches.
\end{definition}

\Ni We import from~\cite{FriedmannLange13a} the well-known fact that the existence of tableaux in the sense defined above
characterizes satisfiability. In~\cite{FriedmannLange13a}, the
result is shown for the more general \emph{unguarded} $\mu$-calculus; we note
that the restriction to guarded formulas does not invalidate the theorem.
\begin{theorem}[\cite{FriedmannLange13a}]\label{thm:sattab}
A $\mu$-calculus formula $\psi$ is satisfiable if and only if there is a tableau for $\psi$.
\end{theorem}
 Given a formula $\phi$, we
 define the alphabet $\Sigma_{\phi}$
to consist of letters that each identify a rule $R\in\mathcal{R}$, a principal formula from $\mathbf{F}$ and one of the conclusions of $R$.  
E.g. the letter $((\vee),0,p\vee\Diamond q)$ identifies the application of the disjunction rule
to a principal formula $p\vee\Diamond q$ and the choice
of the left conclusion; thus this letter
identifies the transition from $p\vee\Diamond q$ to $p$ by use of rule $(\vee)$. We note $|\Sigma_{\phi}|\in\mathcal{O}(|\phi|)$.
Further, we denote the set of all words that encode some branch and some bad branch in some pre-tableau for $\phi$ by 
$\mathsf{Branch}(\phi)$ and $\mathsf{BadBranch}(\phi)$, respectively.

As a crucial result, we now show that limit-deterministic
automata are expressive enough to exactly recognize the
bad branches in pre-tableaux for weakly aconjunctive formulas.

\begin{lemma}\label{lem:tracking_automaton}
Let $\phi$ be a weakly aconjunctive formula. Then
there is a \emph{limit-deterministic} PA $\mathcal{A}=(V,\Sigma_{\phi},\delta,\phi,\alpha)$
with $|V|\leq |\phi|$ and $\mathsf{idx}(\mathcal{A})\leq \mathsf{ad}(\phi)+1$ such that
$L(\mathcal{A})\cap\mathsf{Branch}(\phi)=\mathsf{BadBranch}(\phi)$. 
\end{lemma}
\begin{proof}[Sketch]
The automaton nondeterministically guesses formulas to be tracked, one at a time; 
the set of states of the automaton is the Fischer-Ladner closure of $\phi$.
The priorities of the transitions in the automaton are derived from the alternation level of the target formula
of the respective transition; then every word $w\in L(\mathcal{A})$ that encodes some branch encodes a bad branch. 
Once a deferral is tracked, weak aconjunctivity
implies that all compartments to which the tracked formula belongs are internally deterministic; this is
the case since for conjunctions 
$\psi=\psi_0\wedge \Diamond \psi_1\wedge\ldots \wedge\Diamond\psi_n\wedge\Box(\psi_1\vee\ldots\vee\psi_n)$
-- the only case that can introduce nondeterminism -- 
each next modal step determines just one of the formulas $\psi_i$ that has to be tracked; the conjunct
$\psi_0$ does not contain active $\mu$-variables, so tracking it causes the automaton to leave all compartments to
which $\psi$ belongs. Thus the automaton is limit-deterministic.\qed
\end{proof}

\begin{example} We consider the aconjunctive formula 
\begin{align*}
\phi=\mu X.(\;p\wedge\nu Y.\;(\Diamond (Y\wedge p)\vee \Diamond X))
\end{align*}
which expresses the existence of a finite or infinite path on which $p$ holds everywhere.
We have the $\phi$-deferrals $\phi\epsilon$, $\psi:=(p\wedge\nu Y.\;(\Diamond (Y\wedge p)\vee \Diamond X))\sigma_1$, $\theta:=(\nu Y.\;(\Diamond (Y\wedge p)\vee \Diamond X))\sigma_1$,
$\chi:=(\Diamond (Y\wedge p)\vee \Diamond X)\sigma_2$, $(\Diamond(Y\wedge p))\sigma_2$, $\tau:=(Y\wedge p)\sigma_2$, 
$Y\sigma_2$, $\Diamond X\sigma_2$ and $X\sigma_2$,
where $\sigma_1=[X\mapsto \phi]$ and $\sigma_2=[Y\mapsto\psi];\sigma_1$.
We consider a pre-tableau $P_\phi$ for $\phi$ 
and like in the proof of 
Lemma~\ref{lem:tracking_automaton}, we
construct the limit-deterministic tracking automaton $\mathcal{A}_\phi$, depicted below:\\

\begin{minipage}{.4\linewidth}
\vspace{-35pt}$P_\phi$:
\begin{tiny}
\begin{prooftree}
\alwaysRootAtTop
\AxiomC{{\textbf{7: 3}}}
\LeftLabel{\tiny$(\wedge)$}
\UnaryInfC{{\textbf{6: }$\varsigma$}}
\LeftLabel{\tiny$(\Diamond)$}
\UnaryInfC{{\textbf{5: }$p,\Diamond\varsigma$}}
\AxiomC{{\textbf{9: 1}}}
\LeftLabel{\tiny$(\Diamond)$}
\UnaryInfC{{\textbf{8: }$p,\Diamond \phi$}}
\LeftLabel{\tiny$(\vee$)}
\BinaryInfC{{\textbf{4: }$p,\chi$}}
\LeftLabel{\tiny$(\nu)$}
\UnaryInfC{{\textbf{3: }$p,\theta$}}
\LeftLabel{\tiny$(\wedge)$}
\UnaryInfC{{\textbf{2: }$\psi$}}
\LeftLabel{\tiny$(\mu)$}
\UnaryInfC{{\textbf{1: }$\phi$}}
\end{prooftree}
\end{tiny}
    \end{minipage}%
$\qquad$
    \begin{minipage}{.6\linewidth}
$\mathcal{A}_\phi$:\\
\tikzset{every state/.style={minimum size=15pt}}
\begin{tiny}
  \begin{tikzpicture}[
		auto,
    node distance=1.2cm,
    semithick
    ]
     \node[state,initial above] (e) {$\phi$};
     \node[state] (f) [right of=e] {$\psi$};
     \node[state] (g) [below of=f] {$\theta$};
     \node[state] (h) [below of=g] {$\chi$};
     \node[state] (i) [right of=h] {$\Diamond\varsigma$};
     \node[state] (j) [right of=g] {$\varsigma$};
     \node[state] (k) [left of=g] {$\Diamond\phi$};
     \node[state] (p) [right of=f] {$p$};
     \path[->] (e) edge node [above] {$(\mu),2$} (f);
     \path[->] (f) edge node [left] {$(\wedge),1$} (g);
     \path[->] (f) edge node [above] {$(\wedge),3$} (p);
     \path[->] (g) edge node [right] {$(\nu),1$} (h);
     \path[->] (h) edge node [above] {$(\vee)_l,1$} (i);
     \path[->] (h) edge node [below] {$(\vee)_r,1\;\;\;\;\;\;\;\;$} (k);
     \path[->] (i) edge node [right] {$(\Diamond),1$} (j);
     \path[->] (j) edge node [right] {$(\wedge)$,3} (p);
     \path[->] (j) edge node [above] {$(\wedge)$,1} (g);
  	 \path[->] (k) edge node [left] {$(\Diamond)$,2} (e);

  \end{tikzpicture}
\end{tiny}
\\
\end{minipage} 
The priorities in $\mathcal{A}_\phi$ are derived as follows:
As $\ad(\phi)=2$ is even, we put $k=\ad(\phi)+1=3$; since $\al(\phi)=\al(\psi)=1$, $\alpha(\phi,(\mu),\psi)=\alpha(\Diamond\phi,(\Diamond),\phi)=k-\al(\phi)=2$
and since $\al(p)=0$, $\alpha(\psi,(\wedge),p)=\alpha(\varsigma,(\wedge),p)=k-\al(\phi)=3$. All other formulas have alternation level 2
and transitions to them obtain priority 1.
The tracking automaton accepts exactly those branches in $P_\phi$	that start at node \textbf{1}
and take the loop through node \textbf{9} infinitely often; in these branches, $\phi$ can be
tracked forever and evolves to $\phi$ infinitely often, i.e. their dominating formula is
the least fixpoint formula $\phi$. All other branches 
loop through node \textbf{7} without passing node~\textbf{9} from some point on; 
their dominating fixpoint formula is $\theta$, a greatest fixpoint formula.
We observe that due to the aconjunctivity
of $\phi$, $\mathcal{A}_\phi$ is limit-deterministic since the only two nondeterministic
states $\psi$ and $\varsigma$ each have only one outgoing $(\wedge)$-transition with priority less than $k=3$.

\end{example}

\noindent Given a weakly aconjunctive formula $\phi$, we use Lemma~\ref{lem:tracking_automaton} to construct a limit-deterministic tracking automaton $\mathcal{A}_\phi$
with $L(\mathcal{A}_\phi)\cap\mathsf{Branch}(\phi)=\mathsf{BadBranch}(\phi)$.
Then we put Lemma~\ref{lem:npanba} to use to obtain an equivalent BA in which
all states from $Q=\mathsf{reach}(\pi_3[F])$ are \emph{levelled deferrals}, i.e. pairs $(\psi,q)$ consisting of a deferral $\psi$ and
a number $q\leq \lceil\frac{k}{2}\rceil$, 
the \emph{level} of the pair $(\psi,q)$; the level $q$ encodes the odd alternation level $2q-1$.
A levelled deferral $(\psi,q)$ is \emph{active} if $\mathsf{al}(\psi)=2q-1$ and the automaton
accepts branches which contain a levelled deferral that is active infinitely often without being finished.
The set $\overline{Q}$ is just a subset of $\mathbf{F}$.
Next we use Theorem~\ref{thm:edba} to transform this BA to a DPA $\mathcal{B}_\phi$ with
 $L(\mathcal{A}_\phi)=L(\mathcal{B}_\phi)$. We complement $\mathcal{B}_\phi$ to 
a DPA $\mathcal{C}_\phi=(W,\Sigma_{\phi},\delta,\phi,\alpha)$ by decreasing the priority of each
state in $\mathcal{B}_\phi$ by one; 
we have $L(\mathcal{C}_\phi)=\ov{L(\mathcal{B}_\phi)}$, that is,
$\mathcal{C}_\phi$ accepts exactly those words that encode only `good' branches,
if they encode some branch in some pre-tableau for $\phi$.
By construction,
$|W|\in\mathcal{O}((nk)!)$ and $\mathcal{C}_\phi$ has at most $nk+1$ priorities, and
(recalling Definitions~\ref{defn:permdet} and~\ref{defn:edpatoedba}) the states in the carrier $W$ of $\mathcal{C}_\phi$ 
are of the shape $(U,l)$, where $U$ is a subset of $\mathbf{F}$ and $l$ is a partial permutation of levelled deferrals. For a transition
$t=((U,l),r,(V,l'))$ with $(U,l),(V,l')\in W$, $r\in\Sigma_{\phi}$, if $\alpha(t)=2(n-a)+1$, then~$a$
is the lowest number such that $\mathsf{al}(\phi)=2q-1$, where
$l'(a)=(\phi,q)$ and the $a$-th element of $l$ is not removed by the transition $t$
(i.e. $\alpha(t)$ references the oldest levelled deferral in $l'$ that is active but not removed by the transition $t$) and
if $\alpha(t)=2(n-r)+2$, then $\alpha(t)$ is the index
of the oldest levelled deferral $(\phi,2q-1)$ that is finished (i.e. removed from $l$) in the transition $t$ 
of the automaton $\mathcal{C}_\phi$, which means
that the according $r$-transition in $\mathcal{A}_\phi$ makes $\phi$ leave its 
$2q-1$-compartment. For a state $v=(U,l)$, we define the \emph{label} $\Gamma(v)$ of $v$ as $\Gamma(v)=U$.

\subsection{Permutation Games}

The deterministic parity automaton $\mathcal{C}_\phi$ can now be combined with
applications of tableau rules from~$\mathcal{R}$ to form
a satisfiability game for $\phi$. We proceed to recall the
definition of parity games and some ensuing basic notions.
A \emph{parity game} is a graph $\mathcal{G}=(V,E,\alpha)$ that consists of a set of nodes $V$, a set of edges $E\subseteq V\times V$
and a priority function $\alpha:E\to\mathbb{N}$, assigning priorities to \emph{edges}.
We assume $V=V_\exists\dotcup V_\forall$, that is, every node in $V$
either belongs to player $\mathsf{Eloise}$ ($V_\exists$) or to player $\mathsf{Abelard}$ ($V_\forall$).
A \emph{play} $\rho$ of $\mathcal{G}$ is a (possibly infinite) sequence $v_0v_1\ldots$ such that for
all $i\geq 0$, $v_i\in V$ and $(v_i,v_{i+1})\in E$. 
A play $\rho$ of $\mathcal{G}$ is won by $\mathsf{Eloise}$ 
if and only if $\rho$ is finite and ends in a node that belongs to $\mathsf{Abelard}$ or $\rho$ is infinite and 
$\max(\mathsf{Inf}(\alpha\circ\mathsf{trans}(\rho)))$ is even
 (where $\mathsf{trans}(\rho)$ is defined by $\mathsf{trans}(\rho)(i)=(\rho(i),\rho(i+1))$); $\mathsf{Abelard}$ wins a play $\rho$ if and only if $\mathsf{Eloise}$ does not win $\rho$.
A (memoryless) strategy $s:V\pfun V$ assigns moves to states. A play $\rho$ \emph{conforms} to a strategy $s$ if for
all $\rho(i)\in\mathsf{dom}(s)$, $\rho(i+1)=s(\rho(i))$.
$\mathsf{Eloise}$ has a winning strategy for a node $v$ if
there is a strategy $s:V_\exists\to V$ such that every play of $\mathcal{G}$ that starts at $v$ and conforms to $s$ is
won by $\mathsf{Eloise}$; we have a dual notion of winning strategies for $\mathsf{Abelard}$.
The winning regions $\mathsf{win}_\exists(\mathcal{G})$ and $\mathsf{win}_\forall(\mathcal{G})$
are the sets of those nodes for which $\mathsf{Eloise}$ and $\mathsf{Abelard}$ have winning strategies, respectively.
\emph{Solving} a parity game $\mathcal{G}$ (locally) for a particular node $v\in V$ amounts to
computing the winner of $v$.

Now we are ready to 
define permutation games for
weakly aconjunctive formulas $\phi$,
using the DPA $\mathcal{C}_\phi=(W,\Sigma_{\phi},\delta,\phi,\alpha)$ from the previous section.

\begin{definition}[Permutation games]
Let $\phi$ be a weakly aconjunctive formula.
We define the \emph{permutation game} $\mathcal{G}(\phi)=(W,E,\beta)$ to be a parity game that has
the carrier of $\mathcal{C}_\phi$ as set of nodes. 
For every node $v\in W$ for which $\Gamma(v)$ is not a state, we fix a single rule that is to be applied to $\Gamma(v)$
and a single principal formula $\psi_v\in\Gamma(v)$ to which the rule is to be applied. If $(\vee)$ is to be applied to $\Gamma(v)$, then we put $v\in W_\exists$; otherwise, $v\in W_\forall$. In particular, all state nodes are contained in $W_\forall$.
For $v\in W$, we put $E(v)= \bigcup\{\delta(v,a)\mid a\in \Sigma_v\}$,
where $\Sigma_v\subseteq\Sigma_{\phi}$ consists of all letters $a$ that encode the application of some rule to $\Gamma(v)$ with the condition that the principal
formula of the rule application must be $\psi_v$ if $v$ is not a state node.
Finally, we put $\beta(v,w)=\alpha(v,a,w)$ for $(v,w)\in E$, where $a\in\Sigma_v$ encodes the rule application that
leads from $v$ to $w$.
\end{definition}

\begin{theorem}
Let $\phi$ be a closed, irreducible and weakly aconjunctive formula. Then we have 
$(\{\phi\},[\,])\in\mathsf{win}_\exists(\mathcal{G}(\phi))$
if and only if $\phi$ is satisfiable.
\end{theorem}
\begin{proof}
  By construction, $\mathsf{Eloise}$ wins $(\{\phi\},[\,])$ if and only if there
  is a tableau for $\phi$ (labelled by the labelling function
  $\Gamma$); we are done by Theorem~\ref{thm:sattab}. \qed
\end{proof}

\noindent Due to the relatively simple structure and the
asymptotically smaller size of the determinized automata
$\mathcal{C}_\phi$, the resulting permutation games are somewhat easier to
construct and can be solved asymptotically faster than the structures
created by standard satisfiability decision procedures for the full
$\mu$-calculus (e.g.~\cite{FriedmannLange13a,EmersonJutla99}) which
employ the full Safra/Piterman-construction; note however, that our
method is restricted to the weakly aconjunctive fragment.

\begin{corollary}
The satisfiability of weakly aconjunctive $\mu$-calculus formulas can be decided by solving
parity games of size $\lO((nk)!)$ and $\lO(nk)$ priorities. 
\end{corollary}

\noindent The winning strategies for $\mathsf{Eloise}$ or
$\mathsf{Abelard}$ in these games define models for or  refutations of the respective formulas, so that we have

\begin{corollary}
Satisfiable weakly aconjunctive $\mu$-calculus formulas have models of size $\lO((nk)!)$.
\end{corollary}



\section{Implementation and Benchmarking}\label{section:cool}

We have implemented the permutation satisfiability games 
as an extension of
the \emph{Coalgebraic Ontology Logic Reasoner}
(COOL)~\cite{GorinEA14}, a generic reasoner for coalgebraic modal
logics\footnote{available at
\url{https://www8.cs.fau.de/research:software:cool}}. COOL achieves its
genericity by instantiating an abstract reasoner that works for
all coalgebraic logics to concrete instances of logics. To incorporate
support for the aconjunctive coalgebraic $\mu$-calculus, we have extended the global
caching algorithm that forms the core of COOL to generate and solve
the corresponding permutation games,
with optional \emph{on-the-fly} solving; games are solved using 
either our own implementation of the fixpoint iteration algorithm
for parity games (as in~\cite{BruseEA14}) or
PGSolver~\cite{FriedmannLange09}, which supports
a range of game solving algorithms.
Instance logics implemented in COOL currently include linear-time, 
relational, monotone, and alternating-time
logics, as well as any logics that arise as combinations thereof. In
particular, this makes COOL, to our knowledge, the only implemented
reasoner for the aconjunctive fragments of the alternating-time
$\mu$-calculus and Parikh's game logic.

Although our tool supports the aconjunctive coalgebraic $\mu$-calculus,
we concentrate on the standard relational aconjunctive $\mu$-calculus for experiments,
as this allows us to compare our implementation with the reasoner 
MLSolver~\cite{FriedmannLange10MLSolver}, which constructs satisfiability games using the
Safra/Piterman-construction and hence supports the full relational $\mu$-calculus;
MLSolver uses PGSolver for game solving.

To test the implementations, we devise two series of hard aconjunctive 
formulas with deep alternating nesting of fixpoints. The following formulas encode that each reachable state in a Kripke structure has one of $n$ priorities 
(encoded by atoms $q_i$ for $1\leq i \leq n$)
and belongs to either $\mathsf{Eloise}$ ($q_e$) or $\mathsf{Abelard}$ ($q_a$):
\small
\begin{align*}
\phi_\mathsf{aut}(n) &= \mathsf{AG}(\bigvee_{1\leq i\leq n} (q_i\wedge\bigwedge_{j\neq i}\neg q_j))&
\phi_\mathsf{game}(n)&=\phi_\mathsf{aut}(n)\wedge\mathsf{AG}((q_e\wedge\neg q_a)\vee (\neg q_e\wedge q_a))
\end{align*}
\normalsize
Here we use $\mathsf{AG}\;\psi$ to abbreviate 
$\nu X.\, (\psi \wedge\Box X)$.
Then the non-emptiness regions in parity automata and $\mathsf{Eloise}$'s winning region
in parity games can be specified by the following
\emph{aconjunctive} formulas (where $\hearts\in\{\Diamond,\Box\}$):
\begin{align*}
\phi_\mathsf{ne}(n) &= \eta X_n.\ldots.\nu X_2.\mu X_1. \psi_\Diamond & \psi_\hearts&=\textstyle\bigvee_{1\leq i\leq n}(q_i\wedge\hearts X_i) \\
\phi_\mathsf{win}(n) &= \eta X_n.\ldots.\nu X_2.\mu X_1.\phi_{\mathsf{strat}}(\psi_\hearts) & \phi_{\mathsf{strat}}(\psi_\hearts)&=(q_e\wedge\psi_\Diamond)\vee (q_a\wedge\psi_\Box)
\end{align*}
Furthermore, we define (for $\hearts\in\{\Diamond,\Box\}$)
\begin{align*}
\theta_\hearts(i) & = (q_i\wedge \hearts Y)\vee \textstyle\bigvee_{i<j\leq n}(q_j\wedge\hearts X)\vee \textstyle\bigvee_{1\leq j\leq i}(q_j\wedge\hearts Z)
\end{align*}

\vspace{-25pt}
\begin{figure}
    \centering
    \begin{minipage}{0.45\textwidth}
        \centering
        \begin{tikzpicture}

\begin{semilogyaxis}[
minor tick num=1,
xtick={0,2,4,6,8,10,12,14,16,18,20},
ytick={0.001,0.01,0.1,1,10,100,1000},
yticklabels={$0.001$,$0.01$,$0.1$,$1$,$10$,$100$,$1000$},
every axis y label/.style=
{at={(ticklabel cs:0.5)},rotate=90,anchor=center},
every axis x label/.style=
{at={(ticklabel cs:0.5)},anchor=center},
tiny,
width=6cm,
height=6cm,
transpose legend,
legend columns=2,
legend style={at={(0.5,-0.13)},anchor=north},
ymode=log,
xlabel={value of n},
ylabel={runtime (s)},
xmin=0,
xmax=14,
ymin=0.001,
ymax=100,
legend entries={COOL,COOL on-the-fly,MLSolver,MLSolverOpt}]

%
\addplot[mark=triangle*,mark options={scale=0.9}] table {

 01 0.003314269
 02 0.003662543
 03 0.005884770
 04 0.012555751
 05 0.018308446
 06 0.073683853
 07 0.120555180
 08 0.372444907
 09 0.762148582
 10 2.440859045
 11 2.804672695
 12 17.792399292
 13 28.907151012

};
\addplot[mark=square*,mark options={scale=0.7}] table {

 01	 0.005665382
 02	 0.008152431
 03	 0.043433778
 04	 1.030009846
 05	 6.500200963
 

};
%
\addplot [mark=o, mark options={scale=0.9}] table {
 01 0.004192700
 02 0.018336507
 03 0.067932390
 04 0.465974425
 05 1.286113652
 06 6.466946388
 07 17.997326345
 08 82.524166961

};
\addplot [mark=pentagon, mark options={scale=0.9}] table {
 01 0.003970929
 02 0.015065566
 03 0.044324115
 04 0.233305732
 05 0.614262219
 06 2.433101352
 07 5.602598234
 08 19.497790750
 09 46.674118111

};
\end{semilogyaxis}
\end{tikzpicture}
    \caption{Times for $\neg\theta_1(n)$ (unsatisfiable)}\label{fig:enpa_is_enba}
    \end{minipage}\quad
    \begin{minipage}{0.45\textwidth}
        \centering
        \begin{tikzpicture}

\begin{semilogyaxis}[
minor tick num=1,
xtick={0,2,4,6,8,10,12,14,16,18,20},
ytick={0.001,0.01,0.1,1,10,100,1000},
yticklabels={$0.001$,$0.01$,$0.1$,$1$,$10$,$100$,$1000$},
every axis y label/.style=
{at={(ticklabel cs:0.5)},rotate=90,anchor=center},
every axis x label/.style=
{at={(ticklabel cs:0.5)},anchor=center},
tiny,
width=6cm,
height=6cm,
transpose legend,
legend columns=2,
legend style={at={(0.5,-0.13)},anchor=north},
ymode=log,
xlabel={value of n},
ylabel={runtime (s)},
xmin=0,
xmax=14,
ymin=0.001,
ymax=100,
legend entries={COOL,COOL on-the-fly,MLSolver,MLSolverOpt}]
%
\addplot[mark=triangle*,mark options={scale=0.9}] table {

 01 0.003601202
 02 0.005545315
 03 0.007086871
 04 0.017020367
 05 0.034684609
 06 0.079790132
 07 0.188890539
 08 0.537395370
 09 0.944983616
 10 2.360038499
 11 5.502164128
 12 16.836148465
 13 46.449965578

};
\addplot[mark=square*,mark options={scale=0.7}] table {

 01 0.004997028
 02 0.013581729
 03 0.035175001
 04 0.310558031
 05 0.675793218
 06 6.247625326
 07 33.527906866


};

\addplot [mark=o, mark options={scale=0.9}] table {

 01 0.008368822
 02 0.027447152
 03 0.268239084
 04 0.649396429
 05 3.836329333
 06 6.587703640
 07 27.995679005
 08 44.138908385

};
\addplot [mark=pentagon, mark options={scale=0.9}] table {
 01 0.005812380
 02 0.015746107
 03 0.129773514
 04 0.279059063
 05 1.309766650
 06 2.226867748
 07 8.277059529
 08 12.354744792
 09 39.859538630
 10 57.507097526

};
\end{semilogyaxis}
\end{tikzpicture}
    \caption{Times for $\neg\theta_2(n)$ (unsatisfiable)}\label{fig:domgame}

    \end{minipage}
\end{figure}

\vspace{-15pt}
The following series of valid formulas states that parity automata
with $n$ priorities can be transformed to nondeterministic parity
automata with three priorities without affecting the non-emptiness
region:
\begin{align*}
\theta_1(n):=\phi_{\mathsf{aut}}(n)\to(\phi_{\mathsf{ne}}(n)\leftrightarrow
\textstyle\bigvee_{i\text{ even}}
\mu X.\nu Y.\mu Z.\;\theta_\Diamond(i))
\end{align*}
Similarly, if $\mathsf{Eloise}$ wins a parity game with $n$ priorities,
then she can ensure that in each play, each odd priority $1\leq i\leq n$ is
visited only finitely often, unless a priority greater than $i$ is
visited infinitely often (the converse does not hold in
general~\cite{DawarGraedel08}):
\vspace{-5pt}
\begin{align*}
\theta_2(n):=\phi_{\mathsf{game}}(n)\to(\phi_{\mathsf{win}}(n)\rightarrow
\bigwedge_{i\text{ odd}}
\nu X.\mu Y.\nu Z.\;&\phi_\mathsf{strat}(\theta_\hearts(i))\;)
\end{align*}

\noindent Additionally, we devise two series of unsatisfiable formulas
that exhibit the advantages of COOL's global caching and
on-the-fly-solving capabilities.  These formulas are inspired by the
CTL-formula series $\mathsf{early}(n,j,k)$ and
$\mathsf{early_{gc}}(n,j,k)$ from~\cite{HausmannEA16} but contain
fixpoint-alternation of depth $2^k$ inside the subformula $\theta$:
\vspace{-3pt}
\begin{small}
\begin{align*}
\mathsf{early\text{-}ac}(n,j,k) = & \;\mathit{start}_p \wedge \mathsf{init}(p,n) \wedge \mathsf{init}(r,k) \wedge \mathsf{AG}\;((r \to \mathsf{c}(r,k))\wedge (p \to \mathsf{c}(p,n))) \wedge\\
               & \; \mathsf{AG}\;((\textstyle \bigwedge_{0\leq i\leq j} p_i \to \Diamond(\mathit{start}_r \wedge \theta)) \wedge \neg(p \wedge r)\wedge (r\to \Box\; r))\\
\mathsf{early\text{-}ac}_{\mathsf{gc}}(n,j,k) = & \;\mathsf{early\text{-}ac}(n,j,k) \wedge \mathit{b} \wedge \mathsf{init}(q,n) \wedge \mathsf{AG}\;(\neg(p \wedge q)\wedge \neg(q \wedge r))\wedge\\
        & \;\mathsf{AG}\;((q \to c(q,n)) \wedge \mathsf{AF}\;\mathit{b} \wedge (\mathit{b} \to (\Diamond\;p \wedge \Diamond\;\mathit{start}_q 
                               \wedge \Box\;\neg \mathit{b})))\\
															\mathsf{init}(x,m) = &\; \mathsf{AG}\; ((\mathit{start}_x \to (x \wedge\textstyle \bigwedge_{0\leq i<m} \neg x_i))\wedge (x\to \Diamond\; x))\\
\theta = & \; \eta X_{(2^k)}.\ldots.\nu X_2.\mu X_1.\textstyle\bigvee_{1\leq i\leq 2^k}(\mathsf{bin}(r,i-1)\wedge \Diamond X_i),
\end{align*}
\end{small}

\vspace{-25pt}
\begin{figure}
    \centering
    \begin{minipage}{0.45\textwidth}
        \centering
        \begin{tikzpicture}

\begin{semilogyaxis}[
minor tick num=1,
xtick={0,2,4,6,8,10,12,14,16,18,20},
ytick={0.001,0.01,0.1,1,10,100,1000},
yticklabels={$0.001$,$0.01$,$0.1$,$1$,$10$,$100$,$1000$},
every axis y label/.style=
{at={(ticklabel cs:0.5)},rotate=90,anchor=center},
every axis x label/.style=
{at={(ticklabel cs:0.5)},anchor=center},
tiny,
width=6cm,
height=6cm,
transpose legend,
legend columns=2,
legend style={at={(0.5,-0.13)},anchor=north},
ymode=log,
xlabel={value of n},
ylabel={runtime (s)},
xmin=0,
xmax=20,
ymin=0.001,
ymax=1000,
legend entries={COOL,COOL on-the-fly,MLSolver,MLSolverOpt}]

%
\addplot[mark=triangle*,mark options={scale=0.9}] table {

 01 0.006231262
 02 0.010927707
 03 0.017449466
 04 0.027204472
 05 0.017496431
 06 0.019652914
 07 0.025382195
 08 0.048142076
 09 0.106646225
 10 0.216173437
 11 0.510011437
 12 1.111585836
 13 2.424101644
 14 5.188748123
 15 13.901191444

};
\addplot[mark=square*,mark options={scale=0.7}] table {

01 0.072251309
02 0.092663410
03 0.155972088
04 0.250158983
05 0.323399226
06 0.345304746
07 0.373683738
08 0.410003982
09 0.444002839
10 0.477533792
11 0.510487272
12 0.541856338
13 0.566666153
14 0.594493919
15 0.621080325
16 0.657493417
17 0.700306005
18 0.730832548
19 0.768449998
20 0.800184151

 21	 0.837830630
 22	 0.870720421
 23	 0.913465080
 24	 0.937972216
 25	 0.986445673
 26	 1.018156169
 27	 1.055400456
 28	 1.148819456
 29	 1.131077290
 30	 1.176384372
 31	 1.220402568
 32	 1.270158067
 33	 1.293985419
 34	 1.319552166
 35	 1.362999107
 36	 1.422863366
 37	 1.476002493
 38	 1.509853541
 39	 1.552414880
 40	 1.605384652
 41	 1.654317548
 42	 1.701799244
 43	 1.758973348
 44	 1.817587925
 45	 1.875011226
 46	 1.928383751
 47	 1.978932756
 48	 2.019480723
 49	 2.055387639
 50	 2.127391609
 51	 2.167167451
 52	 2.209224506
 53	 2.296797857
 54	 2.329293598
 55	 2.405896279
 56	 2.459868523
 57	 2.559764744
 58	 2.549671513
 59	 2.614932325
 60	 2.838176642
 61	 2.746929821
 62	 2.810084107
 63	 2.853541808
 64	 2.914404985
 65	 2.969762273
 66	 3.033379548
 67	 3.105339427
 68	 3.155917239
 69	 3.207080923
 70	 3.289633813
 71	 3.314102608
 72	 3.380358064
 73	 3.489473611
 74	 3.503599273
 75	 3.577604240
 76	 3.691695988
 77	 3.728830020
 78	 3.802174391
 79	 3.861474366
 80	 3.940708341
 81	 3.993518972
 82	 4.075693159
 83	 4.142926529
 84	 4.198993755
 85	 4.329940422
 86	 4.375351211
 87	 4.432482685
 88	 4.506050109
 89	 4.617395629
 90	 4.738801723
 91	 4.735750375
 92	 4.826366519
 93	 4.933422190
 94	 4.989241411
 95	 5.074525447
 96	 5.202370427
 97	 5.236857575
 98	 5.387387638
 99	 5.358901072
};
%
\addplot [mark=o, mark options={scale=0.9}] table {
 01 0.427635018
 02 0.960264149
 03 2.837231344
 04 9.485204872
 05 38.215066050
 06 152.670152551

};
\addplot [mark=pentagon, mark options={scale=0.9}] table {
 01 0.023472509
 02 0.022349282
 03 0.025716401
 04 0.028824404
 05 0.034786685
 06 0.041056318
 07 0.056034403
 08 0.089100560
 09 0.163032120
 10 0.335427989
 11 0.687209252
 12 1.474020662
 13 3.301203004
 14 7.273546910
 15 15.902871026
 16 35.088265348
 17 75.090113654
 18 163.131521850
 19 355.587440961
 20 772.415443661

};
\end{semilogyaxis}
\end{tikzpicture}
    \caption{$\mathsf{early}\text{-}\mathsf{ac}(n,4,2)$ (unsatisfiable)}\label{fig:early}
    \end{minipage}\quad
    \begin{minipage}{0.45\textwidth}
        \centering
        \begin{tikzpicture}

\begin{semilogyaxis}[
minor tick num=1,
xtick={0,2,4,6,8,10,12,14,16,18,20},
ytick={0.001,0.01,0.1,1,10,100,1000},
yticklabels={$0.001$,$0.01$,$0.1$,$1$,$10$,$100$,$1000$},
every axis y label/.style=
{at={(ticklabel cs:0.5)},rotate=90,anchor=center},
every axis x label/.style=
{at={(ticklabel cs:0.5)},anchor=center},
tiny,
width=6cm,
height=6cm,
transpose legend,
legend columns=2,
legend style={at={(0.5,-0.13)},anchor=north},
ymode=log,
xlabel={value of n},
ylabel={runtime (s)},
xmin=0,
xmax=20,
ymin=0.001,
ymax=1000,
legend entries={COOL,COOL on-the-fly,MLSolver,MLSolverOpt}]
%
\addplot[mark=triangle*,mark options={scale=0.9}] table {

 01 0.014368858
 02 0.021993548
 03 0.045374779
 04 0.176903013
 05 0.544324350
 06 1.915716061
 07 1.484471835
 08 9.515971747
 09 37.159390910
 10 160.796414006
 11 769.001393999


};
\addplot[mark=square*,mark options={scale=0.7}] table {
  01 0.422572123
  02 0.626352552
  03 1.523379575
  04 4.947750133
  05 1.717908676
  06 11.411818042
  07 5.439501540
  08 3.642606435
  09 15.831136950
  10 17.404179996
  11 18.462094702
  12 6.469337748
  13 9.984823231
  14 6.810827286
  15 6.393029331
  16 5.512443030
  17 6.686021095
  18 7.716896447
  19 19.603414830
  20 30.051589321

 21	 46.274203985
 22	 54.485861510
 23	 57.638827595
 24	 34.882702286
 25	 55.066434156
 26	 90.352377698
 27	 53.693900797
 28	 61.132051868
 29	 13.530712189
 30	 14.813994307
 31	 10.205159622
 32	 10.049598612
 33	 11.124567105
 34	 8.387190052
 35	 18.209036023
 36	 7.937295378
 37	 19.001044045
 38	 19.140983898
 39	 18.975270397
 40	 15.238065888
 41	 21.598725823
 42	 50.042500742
 43	 33.296100292
 44	 54.280548841
 45	 55.422940165

};

\addplot [mark=o, mark options={scale=0.9}] table {
   01	 59.872752244

};
\addplot [mark=pentagon, mark options={scale=0.9}] table {
  
	01 0.799936946
  02 2.125223125
  03 7.764350825
  04 31.856541079
  05 148.177090633
  06 814.145160275

};
\end{semilogyaxis}
\end{tikzpicture}
    \caption{$\mathsf{early}\text{-}\mathsf{ac}_\mathsf{gc}(n,4,2)$ (unsatisfiable)}\label{fig:earlygc}

    \end{minipage}
\end{figure}
\vspace{-15pt}
\noindent where $\mathsf{c}(x,m)$ encodes an $m$-bit counter using atoms $x_0,\ldots,x_{m-1}$ and
$\mathsf{bin}(r,i)$ denotes the binary encoding of the number $i$ using atoms $r_0,\ldots, r_{k-1}$.
The formulas $\mathsf{early\text{-}ac}(n,j,k)$ specify a loop $p$ of length $2^n$ that branches after
$j$ steps to a second loop $r$ of length $2^k$ on which the highest value of the counter
(which counts from $0$ to $2^{k}-1$ and then restarts at $0$) 
is required to be an even number. For constant $k$,
the contradiction on loop $r$
yields a small refutation which can be found early, using on-the-fly solving.
The formulas $\mathsf{early\text{-}ac}_\mathsf{gc}(n,j,k)$ extend this specification by stating
that a third loop $q$ of length $2^n$ is started from loop $p$ infinitely often. Procedures
with sufficient caching capabilities will have to (partially) explore this loop at most once.

We compare the runtimes of MLSolver and COOL on the formulas described
above; 
we let COOL and MLSolver solve games using the local
strategy improvement algorithm \verb|stratimprloc2| provided by PGSolver.
To solve games \emph{on-the-fly}
with COOL however, we use our own implementation of the fixpoint iteration algorithm, which in general
is slower than PGSolver but has the advantage that
it enables on-the-fly solving. With this option enabled, COOL constructs and solves the satisfiability games step by
step and finishes as soon as one of the players has a winning strategy in the partial game.
For COOL, we have conducted all experiments with
and without on-the-fly solving. For MLSolver, we also enabled the
optimizations \verb|-opt litpro| and \verb|-opt comp| (and refer to the resulting prover configuration as MLSolverOpt).
Tests have been run on a system with Intel Core i7 3.60GHz CPU with 16GB RAM.
A more detailed description of the results of the experiments as well as
binaries of a formula generator, the prover COOL and scripts that benchmark the various configurations of the provers
are available in a figshare repository at~\cite{HausmannEArep18}.

We observe that COOL without on-the-fly solving generally finishes faster than both MLSolver and MLSolverOpt
throughout all tested series of formulas (see Figures~\ref{fig:enpa_is_enba}--\ref{fig:earlygc});
the reason for this appears to be that the permutation games solved by COOL are of size $\mathcal{O}((nk)!)$, where $n\leq k$,
and hence asymptotically smaller than the Safra/Piterman games solved by MLSolver which are of size $\mathcal{O}(((nk)!)^2)$.
The size of the refutations for the formulas $\theta_1(n)$ and $\theta_2(n)$ is exponential in $n$ so that
on-the-fly solving does in fact \emph{increase} the runtimes of COOL (see Figures~\ref{fig:enpa_is_enba} and~\ref{fig:domgame}); basically, these formulas
cannot be decided early, and therefore any (necessarily unsuccessful) attempt to do so just consumes additional computation time.
The formulas $\mathsf{early}\text{-}\mathsf{ac}(n,4,2)$ and $\mathsf{early}\text{-}\mathsf{ac}_\mathsf{gc}(n,4,2)$,
on the other hand,
have refutations of size polynomial in $n$, and COOL appears to benefit from on-the-fly solving for these
formulas as it is able to decide them early (see Figures~\ref{fig:early} and \ref{fig:earlygc}). As mentioned above,
COOL uses our own unoptimized implementation of the fixpoint iteration algorithm~\cite{BruseEA14} for on-the-fly solving; 
while this implementation is slower than 
PGSolver's \verb|stratimprloc2| algorithm, the on-the-fly abilities of COOL seem to compensate this disadvantage for 
the $\mathsf{early}\text{-}\mathsf{ac}(n,4,2)$ and $\mathsf{early}\text{-}\mathsf{ac}_\mathsf{gc}(n,4,2)$ formulas
from $n=11$ and $n=8$ on, respectively.


\section{Conclusion}\label{section:conclusion}

We have presented a method to obtain
satisfiability games for the
\emph{weakly aconjunctive} $\mu$-calculus.
The game construction uses
determinization of \emph{limit-deterministic} parity automata, avoiding the full complexity of the
Safra/Piterman construction a) in the presentation of  the procedure and its correctness
proof and b) in the size of the obtained DPA
(which comes from $\lO((nk)!^2)$ to $\lO((nk)!)$).
The resulting permutation satisfiability
games for the weakly aconjunctive $\mu$-calculus
are of size $\lO((nk)!)$, have  
$\lO(nk)$ priorities, and yield a new bound of $\lO((nk)!)$ on the model size for this fragment. We have implemented
this decision procedure in coalgebraic generality and with support for
on-the-fly solving
as part of the coalgebraic satisfiability solver COOL; initial experiments show favourable results.\medskip


The datasets generated and analyzed during the current study are available in the figshare repository: 
 \url{https://doi.org/10.6084/m9.figshare.5919451}

\bibliography{coalgml}

\providecommand{\noopsort}[1]{}
\begin{thebibliography}{10}

\bibitem{BruseEA14}
F.~Bruse, M.~Falk, and M.~Lange.
\newblock The fixpoint-iteration algorithm for parity games.
\newblock In {\em Games, Automata, Logics and Formal Verification, GandALF
  2014}, vol. 161 of {\em {EPTCS}}, pp. 116--130, 2014.

\bibitem{CirsteaEA09}
C.~C{\^i}rstea, C.~Kupke, and D.~Pattinson.
\newblock {EXPTIME} tableaux for the coalgebraic {$\mu$}-calculus.
\newblock In {\em Computer Science Logic, CSL 2009}, vol. 5771 of {\em LNCS},
  pp. 179--193. Springer, 2009.

\bibitem{CourcoubetisY95}
C.~Courcoubetis and M.~Yannakakis.
\newblock The complexity of probabilistic verification.
\newblock {\em J. {ACM}}, 42(4):857--907, 1995.

\bibitem{DawarGraedel08}
A.~Dawar and E.~Gr{\"{a}}del.
\newblock The descriptive complexity of parity games.
\newblock In {\em Computer Science Logic, {CSL} 2008}, vol. 5213 of {\em LNCS},
  pp. 354--368. Springer, 2008.

\bibitem{EmersonJutla99}
E.~A. Emerson and C.~Jutla.
\newblock The complexity of tree automata and logics of programs.
\newblock {\em SIAM J. Comput.}, 29(1):132--158, Sept. 1999.

\bibitem{EsparzaKRS17}
J.~Esparza, J.~Kret{\'{\i}}nsk{\'{y}}, J.~Raskin, and S.~Sickert.
\newblock From {LTL} and limit-deterministic b{\"{u}}chi automata to
  deterministic parity automata.
\newblock In {\em Tools and Algorithms for the Construction and Analysis of
  Systems, {TACAS} 2017}, vol. 10205 of {\em LNCS}, pp. 426--442. Springer,
  2017.

\bibitem{FismanLustig15}
D.~Fisman and Y.~Lustig.
\newblock A modular approach for {B}{\"{u}}chi determinization.
\newblock In {\em Concurrency Theory, {CONCUR} 2015}, vol.~42 of {\em LIPIcs},
  pp. 368--382. Schloss Dagstuhl - Leibniz-Zentrum für Informatik, 2015.

\bibitem{FriedmannLange09}
O.~Friedmann and M.~Lange.
\newblock The {PGSolver} collection of parity game solvers.
\newblock Technical report, LMU Munich, 2009.

\bibitem{FriedmannLange10MLSolver}
O.~Friedmann and M.~Lange.
\newblock A solver for modal fixpoint logics.
\newblock In {\em Methods for Modalities, M4M-6 2009}, vol. 262 of {\em ENTCS},
  pp. 99--111, 2010.

\bibitem{FriedmannLange13a}
O.~Friedmann and M.~Lange.
\newblock Deciding the unguarded modal {\(\mathrm{\mu}\)}-calculus.
\newblock {\em J.\ Appl.\ Non-Classical Log.}, 23:353--371, 2013.

\bibitem{GorinEA14}
D.~Gor{\'in}, D.~Pattinson, L.~Schr{\"o}der, F.~Widmann, and T.~Wi{\ss}mann.
\newblock {COOL} -- a generic reasoner for coalgebraic hybrid logics (system
  description).
\newblock In {\em Automated Reasoning, IJCAR 2014}, vol. 8562 of {\em LNCS},
  pp. 396--402. Springer, 2014.

\bibitem{HausmannEArep18}
D.~Hausmann, L.~Schr\"oder, and H.-P. Deifel.
\newblock Permutation games for the weakly aconjunctive $\mu$-calculus
  (artifact).
\newblock {\em Figshare}, 2018, https://doi.org/10.6084/m9.figshare.5919451.

\bibitem{HausmannEA16}
D.~Hausmann, L.~Schr\"oder, and C.~Egger.
\newblock Global caching for the alternation-free coalgebraic {$\mu$}-calculus.
\newblock In {\em Concurrency Theory, CONCUR 2016}, vol.~59 of {\em LIPIcs},
  pp. 34:1--34:15. Schloss Dagstuhl - Leibniz-Zentrum für Informatik, 2016.

\bibitem{KKV01}
V.~King, O.~Kupferman, and M.~Vardi.
\newblock On the complexity of parity word automata.
\newblock In {\em Foundations of Software Science and Computation Structures,
  FoSSaCS 2001}, vol. 2030 of {\em LNCS}, pp. 276--286. Springer, 2001.

\bibitem{Kozen83}
D.~Kozen.
\newblock Results on the propositional {$\mu$-calculus}.
\newblock {\em Theor.\ Comput.\ Sci.}, 27:333--354, 1983.

\bibitem{Kozen88}
D.~Kozen.
\newblock A finite model theorem for the propositional $\mu$-calculus.
\newblock {\em Stud.\ Log.}, 47:233--241, 1988.

\bibitem{LiuWang09}
W.~Liu and J.~Wang.
\newblock A tighter analysis of {P}iterman's {B}{\"{u}}chi determinization.
\newblock {\em Inf.\ Process.\ Lett.}, 109:941--945, 2009.

\bibitem{NiwinskiWalukiewicz96}
D.~Niwinski and I.~Walukiewicz.
\newblock Games for the {$\mu$}-calculus.
\newblock {\em Theor.\ Comput.\ Sci.}, 163:99--116, 1996.

\bibitem{Piterman07}
N.~Piterman.
\newblock From nondeterministic {B}{\"{u}}chi and {S}treett automata to
  deterministic parity automata.
\newblock {\em Log.\ Meth.\ Comput.\ Sci.}, 3, 2007.

\bibitem{Safra88}
S.~Safra.
\newblock On the complexity of omega-automata.
\newblock In {\em Foundations of Computer Science, FOCS 1988}, pp. 319--327.
  IEEE Computer Society, 1988.

\bibitem{Schewe09}
S.~Schewe.
\newblock Tighter bounds for the determinisation of b{\"{u}}chi automata.
\newblock In {\em Foundations of Software Science and Computational Structures,
  {FOSSACS} 2009}, vol. 5504 of {\em LNCS}, pp. 167--181. Springer, 2009.

\bibitem{ScheweVarghese14}
S.~Schewe and T.~Varghese.
\newblock Determinising parity automata.
\newblock In {\em Mathematical Foundations of Computer Science, {MFCS} 2014},
  vol. 8634 of {\em LNCS}, pp. 486--498. Springer, 2014.

\bibitem{TianDuan14}
C.~Tian and Z.~Duan.
\newblock B\"uchi determinization made tighter.
\newblock {\em CoRR}, abs/1404.1436, 2014.

\bibitem{Walukiewicz00}
I.~Walukiewicz.
\newblock Completeness of {K}ozen's axiomatisation of the propositional
  {$\mu$}-calculus.
\newblock {\em Inf.\ Comput.}, 157:142--182, 2000.

\end{thebibliography}

\appendix

\newpage

\section{Omitted Proofs and Lemmas}

Given an infinite run $\rho$,
we denote by $\mathsf{pre}(\rho,i)$ the run $\rho(0),\ldots,\rho(i)$ and
by $\mathsf{post}(\rho,i)$ the run $\rho(i+1),\rho(i+2),\ldots$.
Similarly, $\mathsf{pre}(w,i)$ and $\mathsf{post}(w,i)$ denote the
words $w(0),\ldots, w(i)$ and $w(i+1),w(i+2)\ldots$, respectively. The
concatenation of two runs $\rho$, $\kappa$ (two words $w_1$, $w_2$) is
denoted by $\rho;\kappa$ ($w_1;w_2$). 

\subsection*{Proof of Theorem~\ref{thm:edba}}

Put $n=|W|$. The number of partial permutations
over the set $Q$ of size $q$ is 
\begin{align*}
|\mathsf{pperm}(Q)|\leq\sum\limits_{i=0}^{q}\frac{q!}{(q-i)!}\leq q! \sum\limits_{i=0}^{\infty}\frac{1}{i!}=q! e;
\end{align*} hence we have
$|W|\leq 2^{n-q}\cdot q!e$. As $q\leq n$, we have that for $n>3$,
$2^{|n-q|}\cdot q!\leq n!$ so that the claimed bound follows. 
For all $t\in\delta'$, we have by definition of $\alpha$ that $1\leq \alpha(t)\leq 2(q-1)+3=2q+1\leq 2n+1$,
i.e. $\mathcal{B}$ has at most $2n+1$ priorities.
It remains to show that $\mathcal{A}$ and $\mathcal{B}$ are equivalent.

Let $w\in L(\mathcal{A})$, i.e. let there be an accepting
run $\sigma\in\mathsf{run}(\mathcal{A},w)$. 
By limit-determinism of $\mathcal{A}$, there is an $i$ s.t. for all $j\geq i$, $\delta|_{\sigma(j),w(j)}\cap Q=\{\mathsf{trans}(\sigma)(j)\}$.
As $\sigma$ is accepting, $\mathsf{Inf}(\mathsf{trans}(\sigma))\cap F\neq\emptyset$. To see that the run 
$\rho=\mathsf{run}(\mathcal{B},w)$ is accepting,
we have to show that the highest priority in $\mathsf{Inf}(\alpha\circ\mathsf{trans}(\rho))$ is even. 
For $j'\geq i$, we have $\rho(j')=(U_{j'},l_{j'})$,
and $Q\ni\sigma(j')=l_{j'}(k)$ for some $k$, i.e. from $i$ on, the corresponding state
from $\sigma$ is contained in the permutation components of $\rho$. 
Since $\sigma$ is infinite (and deterministic from $i$ on), 
$\sigma(j')$ can be tracked forever so that the corresponding position in the permutation component
never changes to $*$ in step 1. of the construction. 
If there is a position $m<k$ for which there is some $j''\geq i$ s.t. $l_{j''}(m)$ changes to $*$ in step 1. 
of the construction, then $k$ is changed to $k-1$, i.e. the considered state moves to the left in the permutation component; this can only happen finitely often without removing the considered state, thus eventually no 
position to the left of the considered state in the permutation changes to $*$;
the considered state is stationary from then on.
Now let $o$ and $k'$ be suitable numbers with $\sigma(o')=l_{o'}(k')$ for all $o'\geq o$;
from $o$ on, $k'$ is the position that tracks the run $\sigma$
and no position with index less than or equal to $k'$ is removed from the permutation component after $o$
and hence for all $q\geq o$, $r_q>k'$, where $r_q$ is the lowest index that turns to $*$ during steps 1. and 2.
of the transition $\mathsf{trans}(\rho)(q)$; i.e. all such transitions with odd priority have priority at most $2(n-k')+1$.
As $\mathsf{Inf}(\mathsf{trans}(\sigma))\cap F\neq \emptyset$, position
$k'$ is active infinitely often, i.e. there are infinitely many $q\geq o$ with $\mathsf{trans}(\sigma)(q)\in F$ and
$\alpha(\mathsf{trans}(\rho)(q))=2(n-k')+2$. Thus $\rho$ is accepting.

Conversely, let $w\in L(\mathcal{B})$, i.e. let the run 
$\rho=\mathsf{run}(\mathcal{B},w)$ be accepting and
let $p$ be the highest priority in $\mathsf{Inf}(\alpha\circ\mathsf{trans}(\rho))$. 
Then $p$ is even and there is some $i$ s.t. for all $j\geq i$, 
$\alpha(\mathsf{trans}(\rho)(j))\leq p$ and there are infinitely many $j\geq i$ with 
$\alpha(\mathsf{trans}(\rho)(j))=p$. Let $\rho(j)=(U_j,l_j)$; abusing notation by interpreting $l_j$ as a set, we observe
 $U_j\cup l_j\subseteq \delta(v_0,\mathsf{pre}(w,j))$, i.e. every state in $U_j$ and $l_j$ can be reached in 
$\mathcal{A}$ from $v_0$ via the word $\mathsf{pre}(w,j)$.
We consider the position $k$ with $2(q-k)+2=p$. For all $j\geq i$ let
$r_j$ denote the lowest index that turns to $*$ in the transition
$((U_j,l_j),w(j),(U_{j+1},l_{j+1}))$; furthermore let $a_j$ denote the
lowest number s.t. $(l_{j}(a_j),w(j),l_{j}(a_j))\in F$ (if no such numbers exists, then
put $r_j=q+1$ and $a_j=q+1$, respectively).
As for all $j\geq i$, $\alpha(\mathsf{trans}(\rho)(j))\leq p=2(q-k)+2$ we always have $r_j>k$
and $a_j\geq k$ since
otherwise $\alpha(\mathsf{trans}(\rho)(j))=2(q-r_j)+3\geq 2(q-k)+3$ or
$\alpha(\mathsf{trans}(\rho)(j))=2(q-a_j)+2>2(q-k)+2$.
Thus the $k$-th component is -- from position $i$ on -- 
stationary in the permutations and infinitely often active. As $l_i(k)\in \delta(v_0,\mathsf{pre}(w,i))$, there is a 
finite run 
$\kappa\in \mathsf{run}_f(\mathcal{A},\mathsf{pre}(w,i))$ with $|\kappa|=i+1$, $\kappa(0)=v_0$ and $\kappa(i)=l_i(k)$.
By definition, $l_i(k)\in Q$ and since $\mathcal{A}$ is limit-deterministic,
there is just a single run $\rho\in\mathsf{run}(\mathcal{A},l_i(k),\mathsf{post}(w,i))$. We note $\rho\subseteq Q$.
Since there are infinitely many $j\geq i$ with $\alpha(\mathsf{trans}(\rho)(j))=2(q-k)+2$, $a_j=k$ and hence $l_j(k)\in F$, we have 
$\mathsf{Inf}(\mathsf{trans}(\rho))\cap F\neq \emptyset$. Thus $\kappa;\rho\in\mathsf{run}(\mathcal{A},w)$ is accepting. \hfill\qed

\subsection*{Proof of Lemma~\ref{lem:npanba}}
The bound on the size of $\mathcal{D}$ follows trivially. Let $w\in L(\mathsf{A})$, i.e. let there
be a run $\rho\in\mathsf{run}(\mathcal{C},w)$ with $\max(\mathsf{Inf}(\alpha\circ\mathsf{trans}(\rho)))$ even.
Then there are $i$ and $l$ s.t. 
$\alpha(\mathsf{trans}(\rho)(i))=2l$, $\alpha(\mathsf{trans}(\rho)(j))\leq 2l$ for all 
$j\geq i$ and there are infinitely many $j\geq i$ with $\alpha(\mathsf{trans}(\rho)(j))=2l$. 
We observe that $\rho(i)\in\delta'(u_0,\mathsf{pre}(w,i))$.
Let $\sigma\in\mathsf{run}(\mathcal{C},\mathsf{pre}(w,i))$ be a witness to this, i.e. a run with $\sigma(i)=\rho(i)$.
We have $(\rho(i+1),l)\in\delta'(\rho(i),w(i))$. As for all $j\geq i$, $\alpha(\mathsf{trans}(\rho)(j))\leq 2l$,
the run $\kappa=\mathsf{run}(\mathcal{D},(\rho(i+1),l),\mathsf{post}(w,i+1))$ is infinite. As 
there are infinitely many $j\geq i$ with $\alpha(\mathsf{trans}(\rho)(j))=2l$, 
$\mathsf{Inf}(\mathsf{trans}(\kappa))\cap F\neq\emptyset$. 
Thus $\sigma;\kappa\in\mathsf{run}(\mathcal{D},w)$ is accepting.

Let $w\in L(\mathcal{D})$, i.e. let there be an accepting run 
$\rho\in\mathsf{run}(\mathcal{D},w)$. Since $\rho$ is accepting, 
there is a number $i$ s.t. $\rho(i)=(v,l)$ for some $l$ and there are 
infinitely many $j\geq i$ with $\mathsf{trans}(\rho)(j)\in F\cap V_{l}$, where
$V_l=\{(v,l)\in W\}$. Since $\rho$ is infinite, 
we have for all $j\geq i$ that $\alpha(v(j),w(j),v(j+1))\leq 2l$, where
$\mathsf{trans}(\rho)(j)=((v(j),l),w(j),(v(j+1),l))$.
Also, there are infinitely many $j\geq i$ with $\mathsf{trans}(\rho)(j)\in F$ and hence $\alpha(v(j),w(j),v(j+1))=2l$, where 
$\mathsf{trans}(\rho)(j)=((v(j),l),w(j),(v(j+1),l))$.
We thus can define a run $\tau\in\mathsf{run}(\mathcal{C},w)$ by putting $\tau(j)=\rho(j)$ for
$j<i$ and $\tau(j)=v_j$ for $j\geq i$ where $\rho(j)=(v_j,l)$. As $\max(\mathsf{Inf}(\mathsf{trans}(\rho)))=2l$,
$\tau$ is accepting.\hfill\qed

\subsection*{Full Proof of Lemma~\ref{lem:tracking_automaton}}
In detail, we put $V=\FLphi$ and recall that any letter from $\Sigma_{\phi}$
identifies a rule application and a conclusion
of the rule application. 
For $\psi\in\FLphi$ and $a=(R,i,\theta)\in\Sigma_{\phi}$, 
we put $\delta(\psi,a)=\{\psi\}$ if $\theta\neq\psi$ and $R\neq(\langle a\rangle)$,
$\delta(\psi,a)=\emptyset$ if $R=(\lightning)$ or $\theta\neq\psi$ and $R\neq(\langle a\rangle)$,
and $\delta(\psi,a)=R(\psi,i)$ if $\theta=\psi$.
Here, $R(\psi,i)$ denotes the set of formulas that $\psi$
changes to when rule $R$ is being applied to it and the $i$-th conclusion
is selected.
If $\ad(\phi)$ is odd, then put $k=\ad(\phi)$,
otherwise put $k=\ad(\phi)+1$; then $k$ is odd.
The priority function $\alpha$ is defined as $\alpha(\psi_1,a,\psi_2)=k-\al(\psi_2)$, for $(\psi_1,a,\psi_2)\in\delta$.
The bounds on the size and index of $\mathcal{A}$ follow.
Since $\phi$ is weakly aconjunctive, all conclusions of rule applications to deferrals 
contain at most one deferral, in particular, for a deferral
$\psi=\psi_0\wedge (\Diamond \psi_1\wedge\ldots \wedge\Diamond\psi_n\wedge\Box(\psi_1\vee\ldots\vee\psi_n))$ with $\mathsf{al}(\psi)=l$, $\delta(\psi,((\wedge),\psi,0))=\{\psi_0,\Diamond \psi_1\wedge\ldots \wedge\Diamond\psi_n\wedge\Box(\psi_1\vee\ldots\vee\psi_n)\}$; between this rule application and the next application of modal rules, we consider $\theta=\Diamond \psi_1\wedge\ldots \wedge\Diamond\psi_n\wedge\Box(\psi_1\vee\ldots\vee\psi_n)$ to be a single compound formula to which no more propositional rules can be applied. Upon the next application of modal rules, each application of a modal rule chooses just one of the $\psi_i$ which needs to be tracked; 
thus we have that for all $a\in\Sigma_{\phi}$, $|\delta|_{\theta,a}|\cap{\alpha_\leq (k)}\leq 1$ and 
since $\psi_0$ contains no active $\mu$-variables
and hence $\mathsf{al}(\psi_0)=0$, $|\delta|_{\psi,a}|\cap{\alpha_\leq (k)}\leq 1$.
Thus $\mathcal{A}$ indeed is limit-deterministic. 
We also have $L(\mathcal{A})\cap\mathsf{Branch}(\phi)=\mathsf{BadBranch}(\phi)$:
To see $L(\mathcal{A})\cap\mathsf{Branch}(\phi)\supseteq\mathsf{BadBranch}(\phi)$, we show that 
$\mathcal{A}$ accepts every bad branch in a pre-tableau for $\phi$. We know that every bad branch induces the list $w\in\Sigma_{\phi}^\omega$ 
of rule applications and selections of conclusions that encode the branch.
Since the branch contains a $\mu$-thread, the automaton can guess the corresponding
formula and follow the single deferral through the thread; this defines a limit-deterministic 
run $\rho\in\mathsf{run}(\mathcal{A},w)$. To see that $\rho$ is accepting, it remains
to show that $\max(\mathsf{Inf}(\alpha\circ\mathsf{trans}(\rho)))$ is even. 
This follows since the tracked thread is a $\mu$-thread, i.e.
we have a formula $\psi$ with odd alternation level $l$ such that $\psi$ occurs infinitely
often in the thread and no formula with lower alternation level than $l$ occurs. As both $k$ and $l$ are odd,
$\max(\mathsf{Inf}(\alpha\circ\mathsf{trans}(\rho)))=k-l$ is even, as required. 
For the converse direction, we have to show that every word that
is accepted by the automaton and encodes a branch encodes a bad branch. So let $w\in L(\mathcal{A})$ encode a branch;
then there is a limit-deterministic accepting run $\rho\in\mathsf{run}(\mathcal{A},w)$, i.e. there 
is some $i$ such that for all $j\geq i$, 
we have $\alpha(\mathsf{trans}(\rho)(j))<l$ for some even $l$ and 
$\delta|_{\rho(j),w(j)}\cap\alpha_\leq{l}=\{\mathsf{trans}(\rho)(j)\}$. We observe
that $\rho(i)$ is a deferral that can be tracked forever through the branch that is encoded by $w$, i.e.
$\rho$ is a thread and $w$ contains $\rho$. 
Let $m=k-l$; as $k$ is odd and $l$ is even, $m$ is odd. Since $\rho$ is accepting, there are infinitely many $j'>j$ with
$\al(\rho(j'))=m$ but there is no $j'>j$ with $\al(\rho(j'))>m$, i.e. the least alternation level
to which the tracked formula evolves infinitely often is odd. Thus $\rho$ is a $\mu$-thread. \qed


\end{document}